\def\year{2022}\relax

\documentclass[letterpaper]{article} %

\usepackage{aaai22}  %
\usepackage{times}  %
\usepackage{helvet} %
\usepackage{courier}  %
\usepackage[hyphens]{url}  %
\usepackage{graphicx} %
\usepackage{natbib}  %
\usepackage{caption} %
\usepackage{subcaption}
\usepackage{amsmath,amsfonts}
\usepackage{amsthm}
\usepackage{algorithm}
\usepackage[noend]{algpseudocode}
\usepackage{bm}
\usepackage{booktabs}
\usepackage{bbm}
\usepackage{enumitem}
\usepackage{chngcntr}
\usepackage{xpatch}

\usepackage{array}
\newcolumntype{R}[2]{%
    >{\adjustbox{angle=#1,lap=\width-(#2)}\bgroup}%
    l%
    <{\egroup}%
}

\usepackage{systeme,mathtools}
\usepackage{tikz}
\usetikzlibrary{shapes,decorations,arrows,calc,arrows.meta,fit,positioning}
\tikzset{
    -Latex,auto,node distance =1 cm and 1 cm,semithick,
    state/.style ={circle, draw, minimum width = 1.0 cm,inner sep=0pt},
    hidden/.style ={circle, dashed, draw, minimum width = 1.0 cm,inner sep=0pt},
    point/.style = {circle, draw, inner sep=0.04cm,fill,node contents={}},
    bidirected/.style={Latex-Latex,dashed},
    el/.style = {inner sep=2pt, align=left, sloped}
}

\DeclareMathOperator*{\argmax}{argmax\,\,}
\DeclareMathOperator*{\maxl}{max\,\,}
\DeclareMathOperator*{\argmin}{argmin\,\,}

\newcommand{\subfiguresize}{0.22\textwidth}
\newcommand{\subfigurefirst}{0.24\textwidth}
\newcommand{\legendsize}{0.85\textwidth}
\newcommand{\alltreatments}{\beta}
\newcommand{\Alltreatments}{\bm{\beta}}

\newtheorem{theorem}{Theorem}
\newtheorem{problem}{Problem}
\newtheoremstyle{case}{}{}{}{}{}{:}{ }{}
\theoremstyle{case}
\newtheorem{case}{Case}

\makeatletter
\xpatchcmd{\  proof}{\topsep6\p@\@plus6\p@\relax}{}{}{}
\makeatother

\newcommand{\threshold}{\theta}
\newcommand{\Activations}{\mathbf{A}}
\newcommand{\activations}{\mathbf{a}}
\newcommand{\Influence}{I}
\newcommand{\influence}{i}

\newcommand{\bftab}{\fontseries{b}\selectfont}
\newcolumntype{P}[1]{>{\arraybackslash}m{#1}}
\newcolumntype{M}[1]{>{\centering\arraybackslash}m{#1}}

\newcommand{\kdd}[1]{#1}

\newcommand{\myspace}{-0.0em}

\newcommand{\aaaipre}[1]{#1}

\def\true{true}
\def\arxiv{false}

\ifx\arxiv\true
\newcommand{\yesappendix}[1]{}
\newcommand{\noappendix}[1]{{\textcolor{blue}{\textit{#1}}}}
\else
\newcommand{\yesappendix}[1]{#1}
\newcommand{\noappendix}[1]{} 
\fi

\title{Heterogeneous Peer Effects in the Linear Threshold Model}

\ifx\arxiv\true
\author{
    Christopher Tran, Elena Zheleva
}
\affiliations{
    Department of Computer Science, University of Illinois at Chicago \\
    Chicago, IL \\
    \{ctran29, ezheleva\}@uic.edu 
}
\else
\author{
    Christopher Tran, Elena Zheleva
}
\affiliations{
    Department of Computer Science, University of Illinois at Chicago \\
    Chicago, IL \\
    \{ctran29, ezheleva\}@uic.edu 
}
\fi

\begin{document}

\maketitle

\begin{abstract}
	The Linear Threshold Model is a widely used model that describes how information diffuses through a social network. According to this model, an individual adopts an idea or product after the proportion of their neighbors who have adopted it reaches a certain threshold. Typical applications of the Linear Threshold Model assume that thresholds are either the same for all network nodes or randomly distributed, even though some people may be more susceptible to peer pressure than others. To address individual-level differences, we propose causal inference methods for estimating individual thresholds that can more accurately predict whether and when individuals will be affected by their peers. We introduce the concept of heterogeneous peer effects and develop a Structural Causal Model which corresponds to the Linear Threshold Model and supports heterogeneous peer effect identification and estimation. We develop two algorithms for individual threshold estimation, one based on causal trees and one based on causal meta-learners. Our experimental results on synthetic and real-world datasets show that our proposed models can better predict individual-level thresholds in the Linear Threshold Model and thus more precisely predict which nodes will get activated over time.
\end{abstract}

\section{Introduction}\label{sec:intro}

Social networks play a vital role in spreading information, ideas, and behaviors.
For example, medical and agricultural innovations can spread through the world~\cite{rogers2010diffusion}, and new products can spread via word of mouth or viral marketing~\cite{kempe2003maximizing}.
Emotions such as happiness~\cite{fowler-bmj08} and hatefulness~\cite{ribeiro2018characterizing} have also been observed to spread through social networks.
These processes, known as \emph{information diffusion}, have traditionally been studied in the social sciences~\cite{granovetter1978}, but more recently have motivated applications in viral marketing~\cite{kempe2003maximizing} and recommender systems~\cite{nikolakopoulos2019}.
Many models exist that capture the diffusion process, including epidemic models~\cite{kermack1927}, voter models~\cite{clifford1973}, the Independent Cascade~\cite{kempe2003maximizing}, and the Linear Threshold Model (LTM)~\cite{granovetter1978}.
In this work, we focus on LTM\@.

According to LTM, an individual is influenced to adopt a product or idea if the proportion of their friends who have already adopted that product or idea is above some threshold.
\aaaipre{
LTM is a popular model used in many settings, such as modeling the spread of ideas~\cite{rogers2010diffusion}, predicting diffusion~\cite{soni-cd19}, and the development of many prominent influence maximization algorithms~\cite{kempe2003maximizing,chen2010,goyal2011,li2018}.
}
However, node thresholds are typically assumed to either be the same for all nodes or randomly distributed. Moreover, there are no methods for estimating individual thresholds ~\cite{talukder2019threshold}, even though different individuals may have different susceptibility to social influence.
Moreover, LTM has not been studied through a causal inference lens, even though LTM captures a causal concept: ``how many friends does it take to buy a product before they \emph{cause} an individual to buy the same product?''
In our work, we seek to address these shortcomings.

We develop two models for estimating individual-level thresholds from data.
Our models reflect real-world scenarios in which some individuals are more easily influenced than others or differently by specific friends.
To address the variety of characteristics and behaviors of individuals, we propose a new concept, \emph{heterogeneous peer effects} (HPEs) in networks, and contrast it with heterogeneous treatment effects for IID data.
While recent work has developed methods for estimating heterogeneous treatment effects in networks, they only consider the network structure as a confounder or as a proxy to latent confounders for individual-level effects and do not estimate heterogeneous peer effects~\cite{veitch-neurips19,guo-wsdm20}.

To facilitate the estimation of peer effects, we develop a Structural Causal Model (SCM) that is specific to LTM and encodes \emph{interference} by \emph{contagion}.
Interference is the influence of treatments or outcomes of peers on an individual, and contagion is the process of friends' outcomes influencing an individual's outcome~\cite{ogburn-stat14}.
For example, a user (e.g., Angelo in Figure~\ref{fig:example_threshold}) might buy sunglasses (outcome) if their friends already bought them (contagion).
Prior work has studied the role of SCMs in the presence of interference~\cite{ogburn-stat14,bhattacharya-uai19} and the estimation of average peer effects~\cite{arbour-kdd16} but not peer effect heterogeneity.
To demonstrate the value of our node threshold prediction algorithms, we evaluate them on three tasks, node threshold estimation, activated node prediction, and diffusion size prediction, using both synthetic and real-world data.

\section{Related work}\label{sec:related}

Diffusion models have been studied for decades in epidemiology~\cite{kermack1927} and opinion dynamics~\cite{holley1975}.
Two popular models for social networks are the Linear Threshold Model (LTM)~\cite{granovetter1978} and Independent Cascade (IC) Model~\cite{goldenberg2001}.
These models have been used in diffusion prediction~\cite{soni-cd19} and influence maximization~\cite{kempe2003maximizing}.
While some work has focused on estimating parameters of IC models~\citep{saito2008,bourigault2016,kalimeris2018}, not much work has explored threshold estimation in the Linear Threshold Model.

A recent survey on influence maximization techniques using LTM highlighted the lack of threshold estimation models~\cite{talukder2019threshold}.
Goyal et al.\ learn edge diffusion probabilities but do not learn individual thresholds~\citep{goyal2010}.
Recent work identified a problem in threshold estimation called the ``opacity problem''~\citep{berry-social19}.
The opacity problem states that the thresholds estimated from data are upwardly biased and propose only to use nodes whose thresholds are \emph{precisely measured}.
They use those thresholds in a regression model to estimate thresholds for all nodes.
In contrast, we estimate individual thresholds on ``snapshots'' of networks.

Our work focuses on defining and identifying heterogeneous effects of peers through the use of a Structural Causal Model (SCM)~\cite{pearl-book09} to estimate individual thresholds in the LTM.
Shalizi and Thomas explore causal inference under the presence of homophily and contagion~\cite{shalizi2011homophily}.
They address the problem of identification of social effects when influence is decoupled from group effects~\cite{manski-social93}.
Ogburn and VanderWeele presented an extensive discussion of SCMs for interference~\cite{ogburn-stat14}.
Arbour et al.\ utilize abstract ground graphs for estimating average effects in social networks~\cite{arbour-kdd16}.
To the best of our knowledge, there have been no SCMs developed for the LTM.
In addition, recent work has introduced methods for estimating heterogeneous \emph{treatment} effects in networks but do not estimate effects of peers~\cite{guo-wsdm20}.

\section{Problem Description}\label{sec:problem_setup}

To define the problems of individual threshold estimation and conditional average peer effect estimation, we first present the data, linear threshold model, and define heterogeneous peer effects.

\subsection{Data Model}\label{sec:diffusion}

\begin{figure}
    \centering
    \includegraphics[width=0.65\columnwidth]{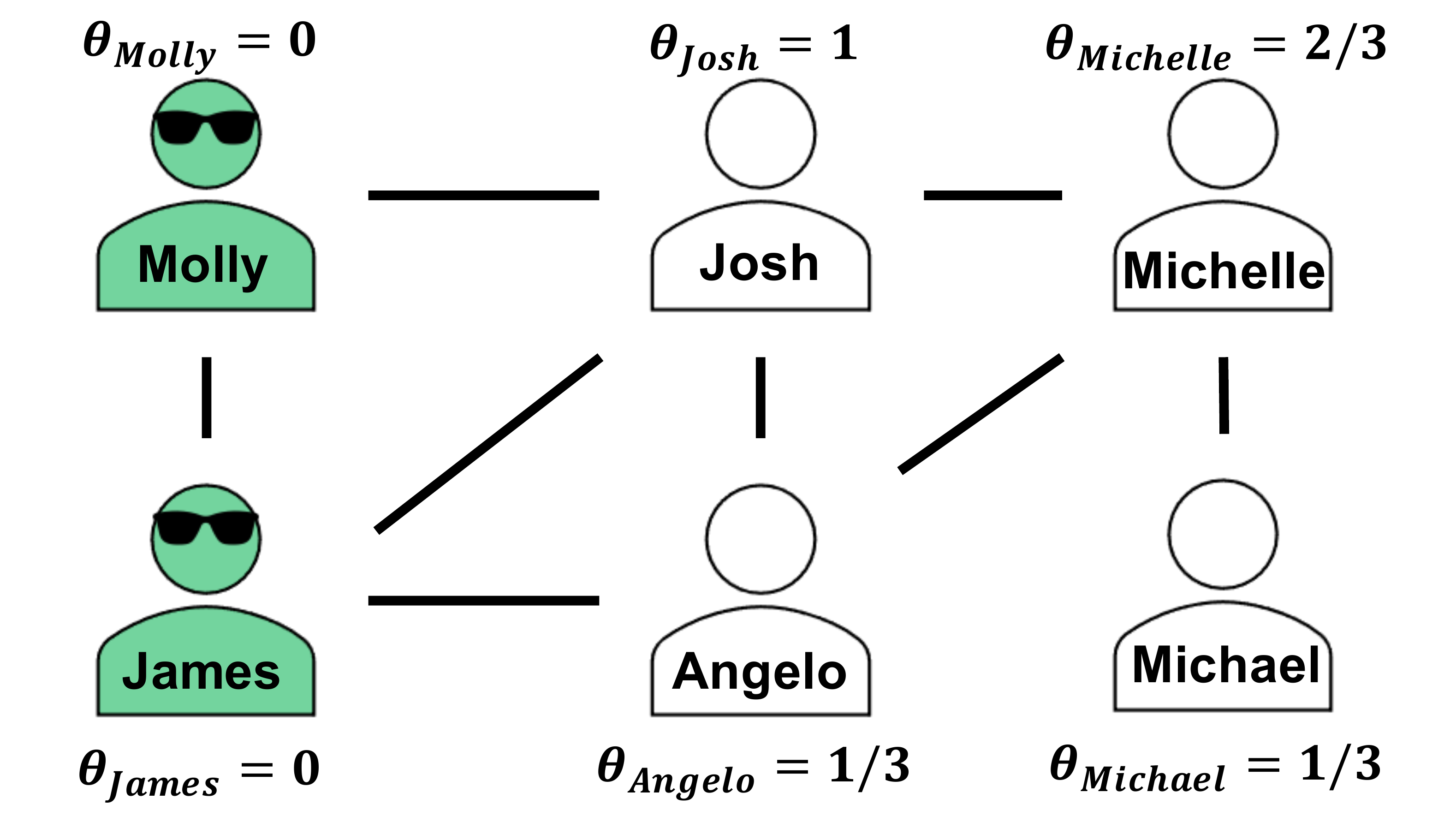}
    \caption{A toy example illustrating the Linear Threshold Model diffusion process, where Molly and James are activated. In the first time step, Angelo will become activated, since his threshold, $\threshold_{Angelo}=1/3$.}\label{fig:example_threshold}
    \vspace{-1.0em}
\end{figure}

Let $G=(\mathbf{V}, \mathbf{E})$ denote an attributed social network, where $\mathbf{V}$ is the set of nodes and $\mathbf{E}$ is the set of edges between nodes.
If two nodes $v, u \in \mathbf{V}$ are connected by an edge, then $(v, u) \in \mathbf{E}$, denotes an edge from $v$ to $u$.
Define the set of neighbors of $v$ to be $N(v)=\{u: u \in \mathbf{V}, (u, v) \in \mathbf{E} \}$.
Each node $v \in \mathbf{V}$ has an observed $m$-dimensional vector of attributes, $\mathbf{X}_v$, unobserved characteristics, $\mathbf{U}_v$, and an outcome of interest $Y_v \in \{0, 1\}$, which is a binary indicator of whether the node is activated (e.g., \ whether an individual has bought new sunglasses).
We define the set of activated nodes at time $t$ to be $\mathcal{D}_t = \{ v : Y_v=1\}$.

\subsection{Linear Threshold Model}\label{sec:linear_threshold}

According to the Linear Threshold Model (LTM), each node $v$ has an activation threshold $\threshold_v$.
Given an initial set of activated nodes, $\mathcal{D}_0  \subseteq \mathbf{V}$, diffusion occurs in discrete steps, $t=1, 2, \dots, T$.
In each time step $t$, a node $v \in \mathbf{V} \setminus \mathcal{D}_i$ is activated if the \emph{activation influence}, the weighted proportion
of its activated neighbors, reaches the node's threshold $\threshold_v$:
\begin{equation}
\label{eq:threshold}
  \sum_{u \in N(v)} w_{uv} Y_u^t \geq \threshold_v,
\end{equation}
where $w_{uv}$ is the normalized influence weight of neighbor $u$ on $v$.
According to LTM, nodes can only become activated as activated neighbors increase.
In practice, node thresholds are implemented by considering random or uniform thresholds~\cite{talukder2019threshold} even though the propensity to be influenced can vary from individual to individual.
Our goal in this work is to estimate thresholds for all nodes.
\begin{problem}\label{prob:hte_ltm}
	(Individual node threshold estimation for LTM) 
	Let $G=(\mathbf{V}, \mathbf{E})$ be a graph and $\bm{\threshold} = \{\threshold_v \mid v \in \bm{V}\}$ be a set of node thresholds. The goal is to estimate the thresholds for all nodes, $\bm{\hat{\threshold}} = \{\hat{\threshold}_v \mid v \in \mathbf{V}\}$, such that the average mean squared error between $\theta$ and $\hat{\theta}$ is minimized:
	\begin{equation}
    	\argmin_{\hat{\theta}_v} \text{MSE} = \frac{1}{|V|} \sum_{v}^{|V|} \Big(\theta_v - \hat{\theta}_v \Big)^2.
    \end{equation}
\end{problem}

Figure~\ref{fig:example_threshold} shows a toy example of a social network with five individuals.
Each node has their own threshold (e.g., $\threshold_{\text{Angelo}} = 1/3$ means Angelo's threshold is $1/3$).
The initial set of activations are the individuals who have adopted a product (sunglasses), which consists of two individuals: $\mathcal{D}_0$ = \{Molly, James\}.
Assuming equal weights, Angelo will buy new sunglasses in the first time step since one of his three friends (James) has already bought them.
No one else will buy sunglasses in subsequent steps since Josh's threshold is $1$ and Michelle's threshold is $2/3$.

\subsection{Causal Inference in LTM}

Here, we connect causal inference with the Linear Threshold Model under the \aaaipre{potential outcomes framework} and define heterogeneous peer effects and their estimation.

\subsubsection{Causal inference in networks}

A common assumption in estimating causal effects is the stable unit treatment value assumption (SUTVA)~\cite{rubin-1978}, which is the assumption that the outcome of an individual is independent of the treatment assignment of other individuals: $Y_v(Z_v) \perp Z_u, \forall u \neq v \in \bm{V}$.
However, this assumption is violated in network data because of the interaction between individuals - known as \emph{interference}~\cite{ogburn-stat14} - which can lead to biased causal effect estimations.
Interference occurs in the LTM, where activation depends on the activation of neighbors. 
Thus, we need to account for node features \emph{and} network interference when estimating effects.

Define $\Activations = \{Z_v \mid v \in \mathbf{V}\}$ to be the set of treatment variables for all nodes in the network.
Then, the outcome of a node is dependent on two variables: the individual treatment and the set of treatments in the network: $Y_v(Z_v, \mathbf{A})$.
We can define the average treatment effect (ATE) and the average peer effect (APE)~\cite{arbour-kdd16,fatemi-icwsm20} as:
\begin{gather}
  \text{ATE} = E[Y(Z=1) -  Y(Z=0) \mid \Activations =  \activations],
  \\
  \text{APE} =  E[Y(\Activations =  \activations) - Y(\Activations =  \activations') \mid Z=z].
\end{gather}
\noindent ATE estimates the effect of treatment on a node, keeping other treatment assignments the same. 
APE estimates the effect when keeping a node's treatment fixed while changing treatments of other nodes.

\emph{Contagion} is a case of \emph{interference}, where peer outcomes in the previous time step affect individual outcomes.
\aaaipre{
In the LTM, the outcome of peers in the previous time step are captured by  $\Activations = \{Y_v^{t-1} \mid v \in \mathbf{V} \}$, which is the treatment variable of interest.
}
Neighborhood-level variables (e.g. $\Activations$) are \emph{relational} and individual-level variables (e.g. $\mathbf{X}, Y, Z$) are \emph{propositional} variables.

\subsubsection{Heterogeneous peer effects}\label{sec:hpe}

Heterogeneous treatment effects (HTEs) occur when subpopulations react differently to treatment.
HTE is expressed through the conditional average treatment effect (CATE).
CATE is defined as the expected difference in potential outcomes, conditional on an individual's features: \( \tau(\mathbf{x}) \equiv E[Y(z) - Y(z') \mid \mathbf{X} = \mathbf{x}] \), which results in personalized effect estimations.
In LTM, this is the expected change in activation if the node was activated vs.\ not activated in the previous time point with features \( \mathbf{X}_v = \mathbf{x} \).
So far, we have ignored the influence of neighbors.
Next, we describe how to incorporate it into the HTE estimation.
First, we define CATE under interference as:
\begin{align}
  \tau(\mathbf{x}) = E[Y(Z=1) - Y(Z=0) | \mathbf{X}=\mathbf{x}, \Activations=\activations].
\end{align}

Different individuals may have different peer effects based on personal characteristics, such as personality and susceptibility to influence. 
To capture this idea, we define \emph{conditional average peer effect (CAPE)}:
\begin{equation}\label{eq:cape}
  \rho(\mathbf{x}) = E[Y(\Activations = \activations) - Y(\Activations = \activations') \mid \mathbf{X}=\mathbf{x}, Z=z]
\end{equation}
\emph{Heterogeneous peer effects (HPEs)} occur when there are at least two possible assignments of $\mathbf{X}$, $\mathbf{x}_i$ and $\mathbf{x}_j$, for which $\rho(\mathbf{x}_i)\neq \rho(\mathbf{x}_j)$. 
Since not all nodes in the network affect every other node, we use neighbors of nodes for contagion.
We can define \( \Activations_v = \{ Y_u^{t-1} \mid u \in N(v) \} \), the set of outcomes of neighbors of \( v \) in the previous time step, which affects the outcome and APE in the current time step.

\begin{problem}\label{prob:cape}
  (Conditional average peer effect estimation)
  Let \( G = (\mathbf{V}, \mathbf{E}) \) be a graph.
Let $\hat{\rho}(\mathbf{X})$ be an estimate for the CAPE.
  The goal is to minimize the expected error between the estimate and its true value:
  \begin{equation}
    \min_{\hat{\rho}} E[(\rho(\mathbf{X}) - \hat{\rho}(\mathbf{X}))^2].
  \end{equation}
\end{problem}
In general, it is not trivial to estimate CAPE since relational variables are sets, and $\Activations_v$ for $v \in \mathbf{V}$ are different sizes.
To circumvent this problem, we model relational variables through aggregate functions, such as the mean or variance~\cite{arbour-kdd16}.
Next, we discuss the link between CAPE estimation and threshold estimation.

\section{HPE and Threshold Estimation}\label{sec:methodology}

We begin by linking CAPE estimation to threshold estimation.
Second, we introduce a causal model to identify the CAPE.
Then, we map the CAPE estimation to a problem of estimating triggers for heterogeneous effects, which allows the estimation of node thresholds for LTM.

\subsection{Causal model for LTM}\label{sec:hpe_ltm}

In LTM, we define the relational variable at time $t$ as the \emph{activation influence} on the left side of Equation~\ref{eq:threshold}:
\begin{equation}
  \label{eq:influence}
  \Influence_v^t = \sum_{u \in N(v)} w_{uv}Y_u^{t-1}.
\end{equation}
We define the treatment variable, \( Z_v \), as the previous outcome: \( Z_v = Y_v^{t-1} \) and the CAPE is defined with the aggregation function in~\eqref{eq:influence}:
\begin{align}
  \rho(\mathbf{x}) = & E[Y(\Influence^t = \influence^t)  - Y(\Influence^t = \influence^{t'}) \mid \mathbf{X} = \mathbf{x}, Z=z],
\end{align}
where $\influence^t$ and $\influence^{t'}$ correspond to two different values of $\mathbf{A}$. 
In order to identify this effect, we need to account for correct \emph{adjustment variables}.
Hence, we now introduce a causal model for the LTM process to identify CAPE for LTM.

\begin{figure}
\centering
\resizebox{0.85\columnwidth}{!}{
  \begin{tikzpicture}[trim left=-0.5cm]
    \node[hidden] (u1) at (0,0.66) {{$\mathbf{U}_v$}};
    \node[state] (x1) at (1.5,0) {{$\mathbf{X}_v$}};
    \node[state] (y11) at (3.5,0) {{$Y_v^0$}};
    \node[state] (y12) at (5.5,0) {{$Y_v^1$}};
    \node[state] (y13) at (7.5,0) {{$Y_v^2$}};
    \node[state] (y1T) at (9.5,0) {{$Y_v^T$}};
    \node[state] (theta) at (1.5,1.25) {{$\theta_v$}};

    \node[state] (xN) at (1.5,2.5) {{$\mathbf{X}_N$}};
    \node[state] (yN1) at (3.5,2.5) {{$\Influence_v^0$}};
    \node[state] (yN2) at (5.5,2.5) {{$\Influence_v^1$}};
    \node[state] (yN3) at (7.5,2.5) {{$\Influence_v^2$}};
    \node[state] (yNT) at (9.5,2.5) {{$\Influence_v^T$}};

    \path (x1) edge (y11);
    \path (y11) edge (y12);
    \path (y12) edge (y13);
    \path (y13) edge[dashed] (y1T);

    \path (theta) edge (y11);
    \path (theta) edge[out=0, in=150] (y12);
    \path (theta) edge[out=0, in=150] (y13);
    \path (theta) edge[out=0, in=150] (y1T);

    \path (xN) edge (yN1);

    \path (yN1) edge (yN2);

    \path (yN2) edge (yN3);
    \path (yN3) edge[dashed] (yNT);

    \path(y11) edge (yN2);

    \path(y12) edge (yN3);
    \path(y13) edge[dashed] (yNT);

    \path(yN1) edge (y12);

    \path(yN2) edge (y13);

    \path(yN3) edge[dashed] (y1T);

    \path(x1) edge[out=-30, in=-150] (y12);
    \path(x1) edge[out=-30, in=-150] (y13);
    \path(x1) edge[out=-30, in=-150] (y1T);

    \path(xN) edge[out=30, in=150] (yN2);

    \path(xN) edge[out=30, in=150] (yN3);
    \path(xN) edge[out=30, in=150] (yNT);

    \path(u1) edge (x1);
    \path(u1) edge (theta);

  \end{tikzpicture}
}
\caption{A causal model of peer effects for a node $v$ that reflects the LTM process.}
\label{fig:causal_model}
\vspace{-1.0em}
\end{figure}
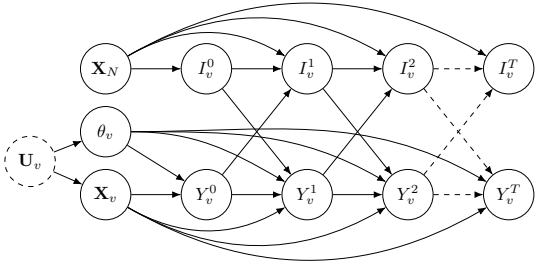

To estimate CAPE in LTM, we develop a Structural Causal Model (SCM) of diffusion that corresponds to the LTM process.
SCMs are graphical models that encode cause-effect relationships between variables and allow for the identification of effects~\cite{pearl-book09}.

Figure~\ref{fig:causal_model} shows the causal model of diffusion we develop to capture the LTM process.
Here, $\mathbf{X}_N$ represents a \emph{set} of neighbor features, and $\Influence_v^t$ is the random variable of \emph{activation influence} on the outcomes of $v$'s neighbors.
An arrow from one variable to another denotes a cause-effect relationship.
Contagion from node $v$'s neighbors are captured through the arrows from $\Influence_v^t$ to $Y_v^{t+1}$.

In order to \emph{identify} the effect of activation influence on a node's own activation (e.g., the effect of $\Influence_v^{1}$ on $Y_v^{2}$, marked in blue in Figure~\ref{fig:causal_model}), we need to find the correct \emph{adjustment set} that blocks all back-door paths, based on the back-door criterion~\cite{pearl-book09,arbour-kdd16}.
In our example, three arrows are going into the treatment $\Influence_v^{1}$: $Y_v^{0}$, $\Influence_v^{0}$ and $\mathbf{X}_N$, opening potential backdoor paths.
Based on this, a sufficient set to block all backdoor paths is the set of $\mathbf{X}_v$, the previous outcome $Y_v^{1}$, and the threshold $\threshold_v$.
In the CAPE formula defined in Equation~\eqref{eq:cape}, this is the heterogeneous subgroup $\mathbf{X} = \mathbf{x}$ that we need to consider.
The special case is the initial peer effect of $\Influence_v^0$ on $Y_v^1$, where there are no backdoor paths from $\Influence_v^0$ to $Y_v^1$.
In this case, we can estimate the effect of $\Influence_v^0$ without adjustment.

To match the characteristics of LTM, we define a functional form of the outcome of node $v$.
The outcome $Y_v^{t+1}$ can be represented as a function of $\big( \mathbf{X}_v, Y_v^{t}, \Influence_v^{t}, \threshold_v \big)$. 
We can define the functional form of $Y_v^{t+1}$ for $t \geq 1$ as:
\begin{equation}
  Y_v ^ {t+1} (\mathbf{X}_v, Y_v^{t}, \Influence_v^{t}, \threshold_v)=
      \begin{cases}
      Y_v^{t} & \text{if $Y_v^{t} = 1$,} \\
      \mathbbm{1} \big[ \Influence_v^{t} \geq \threshold_v \big] & \text{if $Y_v^{t} = 0$.}
      \end{cases}
      \nonumber
  \end{equation}
This form correctly captures the LTM process: a node stays activated if already activated. Otherwise, it activates based on its threshold and the activations of neighbors.
theore\subsection{Individual threshold estimation for LTM}

To estimate CAPE in LTM, we estimate the threshold $\theta$ assuming that $\Influence^t$ is known.
We first map the problem of estimating thresholds to the problem of estimating triggers~\cite{tran-aaai2019} in the context of heterogeneous peer effects.
This mapping allows us to adapt causal trees to the problem of estimating thresholds and opens avenues for further algorithm development.

Since the peer influence, $\Influence_v^t$, is a continuous value, it cannot be modeled using traditional causal inference methods which consider binary treatment values.
We reformulate the peer effect as an effect of a binary treatment: the expected difference when a node's activation influence, $\Influence_v^t$, is above and below its threshold, $\theta_v$.
Then our goal is transformed into estimating the correct node threshold that correctly identifies when a node is above and below its actual threshold.
To do this, we map the threshold estimation problem to the problem of estimating triggers for heterogeneous effects~\cite{tran-aaai2019}.

A \emph{trigger} is defined as the minimum amount of treatment necessary to change an outcome~\cite{tran-aaai2019}.
Some examples of triggers are the minimum number of days a patient needs to take a medicine to be cured or a minimum discount needed for a customer to buy a product.
For the problem of threshold estimation, our causal question is: ``what is the minimum number of activated neighbors that can cause a node with certain attributes to become activated''?
Define a trigger, $\hat{\theta}$, with sets of \aaaipre{potential outcomes} above and below the trigger:
\begin{gather}
    \mathcal{Y}^t(\Influence^t \geq \hat{\theta}) = \{ Y_v^t \mid \Influence_v^t \geq \hat{\theta} \}
    \\
    \mathcal{Y}^t(\Influence^t < \hat{\theta}) = \{ Y_v^t \mid \Influence_v^t < \hat{\theta} \}
\end{gather}
Then, we can define the average peer effect with a global trigger, $\hat{\theta}$ as: $E[\mathcal{Y}^t(\Influence^t \geq \hat{\threshold}) - \mathcal{Y}^t(\Influence^t < \hat{\threshold})]$, which is the case when a global threshold is defined for LTM.

Since we are interested in individual node thresholds, we estimate \emph{heterogeneous triggers} rather than global triggers.
We can now define CAPE with a trigger:
\begin{equation}
\label{eq:trigger_cate}
    \mathcal{P}(\mathbf{x}) = E[\mathcal{Y}(\Influence^t \geq \hat{\threshold}) - \mathcal{Y}(\Influence^t < \hat{\threshold}) \mid \mathbf{X} = \mathbf{x}, Z=z].
\end{equation}
Estimating CAPE with a trigger translates to estimating the effect of influence crossing the trigger.
The implication of mapping the node threshold estimation to a trigger-based heterogeneous treatment effect estimation problem is that an accurate trigger is the best estimate of the node threshold.
\aaaipre{
\begin{theorem}\label{thm:trigger_theorem}
    For any node $v \in \bm{V}$, let $\threshold_v$ be $v$'s true threshold. Then $\hat{\threshold}_v$ that maximizes the CAPE with a trigger %
    \begin{equation}
        \argmax_{\hat{\theta}_v} E[\mathcal{Y}(\Influence^t_v \geq \hat{\threshold}_v) - \mathcal{Y}(\Influence^t_v < \hat{\threshold}_v) \mid \mathbf{X}_v, Z_v],
    \end{equation}
    provides the best estimate of the node threshold, $\threshold_v$.
\end{theorem}
}
\yesappendix{
\noindent The proof is in the Appendix\footnote{https://github.com/edgeslab/hpe-ltm}.
}
\noappendix{
\begin{proof}
    Suppose $|N(v)|=n$ and $w_{uv}$ be an arbitrary weight such that $\sum_{u} w_{uv} = 1$.
	Let the true threshold of node $v$ be $\theta_v$.
    Define $\mathbf{A_v} = \mathbf{a_i}$ to be an assignment of activations of neighbors of $v$ and $W_v(\alpha_i)$ to be the outcome for the activated set $\mathbf{a_i}$:
    \begin{align}
		W_v(\alpha) & = 
    	\begin{cases}
    		0 & \text{if $\Influence_v(\mathbf{a_i}) < \theta_v$}, \\
            1 & \text{if $\Influence_v(\mathbf{a_i}) \geq \theta_v$}.
		\end{cases}
    \end{align}
    Where $\Influence_v$ is the activation influence. 
    Let the set of potential outcomes, $\mathcal{Y}_v$, below and above the estimated trigger $\hat{\theta_v}$ be:
	\begin{align}
		\mathcal{Y}_v(\Influence_v < \hat{\theta}_v) & = \{ W_v(\mathbf{a_i}) \mid \Influence_v(\mathbf{a_i}) < \hat{\theta_v} \}
		\\
		\mathcal{Y}_v(\Influence_v \geq \hat{\theta}_v) & = \{ W_v(\mathbf{a_i}) \mid \Influence_v(\mathbf{a_i}) \geq \hat{\theta}_v \}.
    \end{align}
    Let $N_0$ and $N_1$ be the size of the sets $\mathcal{Y}_v(\Influence_v < \hat{\theta}_v)$ and $\mathcal{Y}_v(\Influence_v \geq \hat{\theta}_v)$, respectively.
    Define the set of outcomes when node $v$ is not activated and activated with the true threshold as $\mathcal{Z}_v(0)$ and $\mathcal{Z}_v(1)$:
    \begin{align}
        \mathcal{Z}_v(0) &= \{W_v(\mathbf{a_i}) \mid \Influence_v(\mathbf{a_i}) < \theta_v \},
        \\
        \mathcal{Z}_v(1) &= \{W_v(\mathbf{a_i}) \mid \Influence_v(\mathbf{a_i}) \geq \theta_v \}.
    \end{align}
    Let $N_{\mathcal{Z}_0}$ and $N_{\mathcal{Z}_1}$ be the cardinalities of set $\mathcal{Z}_v(0)$ and $\mathcal{Z}_v(1)$, respectively.
    The difference between $\mathcal{Y}_v$ and $\mathcal{Z}_v$ is that $\mathcal{Y}_v$ contains the true potential activations based on the \emph{estimated} threshold and $\mathcal{Z}_v$ contains the true potential activations based on the \emph{true} threshold.
    We define expectation over the sets $\mathcal{Y}_v$ and $\mathcal{Z}_v$ as the mean over the set.
    There are three cases for the estimated threshold, $\hat{\theta}_v$:
    \begin{case}[$\hat{\theta}_v < \theta_v$]
        We compute the expected outcomes below and above the estimated trigger $\hat{\theta}_v$:
        \begin{align}
            E[\mathcal{Y}_v(\Influence_v < \hat{\theta}_v)]  & = E[\{ W_v(\mathbf{a_i}) \mid \Influence_v(\mathbf{a_i}) < \hat{\theta}_v\}] \nonumber
            \\
            &= \frac{1}{N_0} \sum_{i} W_v(\mathbf{a_i}) = 0,
        \end{align}
        since there are no times when $W_v$ will be $1$ (activated) since all outcomes below the estimated threshold $\hat{\theta}_v$ are also below the true threshold $\theta_v$.
        The expected value above the estimated threshold is:
		\begin{gather}
            E[\mathcal{Y}_v(\Influence_v  \geq \hat{\theta}_v)]   = E[\{ W_v(\mathbf{a_i}) \mid \Influence_v(\mathbf{a_i}) \geq \hat{\theta}_v\}] \nonumber
            \\
            = \frac{1}{N_1} \sum_i W_v(\mathbf{a_i}) = \frac{1}{N_1} \sum_i W_v(\mathbf{a_i}) + \frac{1}{N_1} \sum_j W_v(\mathbf{a_j}) \nonumber
            \\
            = \frac{1}{N_1} \cdot N_{\mathcal{Z}_1} = \frac{N_{\mathcal{Z}_1}}{N_1}.
        \end{gather}
        The effect is the difference above and below the estimated trigger:
        \begin{align}
            E[\mathcal{Y}_v(\Influence_v \geq \hat{\theta_v}) - \mathcal{Y}_v(\Influence_v < \hat{\theta_v})] & = \frac{N_{\mathcal{Z}_1}}{N_1} - 0 = \frac{N_{\mathcal{Z}_1}}{N_1} < 1.
        \end{align}
        Note that $N_{\mathcal{Z}_1} < N_1$ since $N_1$ has cases above the estimated threshold, which is smaller than the true threshold.
    \end{case}
    \begin{case}[$\hat{\theta}_v > \theta_v$]
        \begin{align}
            E[\mathcal{Y}_v(\Influence_v < \hat{\theta}_v)]  & = E[\{ W_v(\mathbf{a_i}) \mid \Influence_v(\mathbf{a_i}) < \hat{\theta}_v\}]
            \nonumber \\
            & = \frac{1}{N_0} \sum_{i} W_v(\mathbf{a_i}) + \frac{1}{N_0} \sum_{j} W_v(\mathbf{a_i})
            \nonumber \\
            & = 0 + \frac{1}{N_0} \cdot (N_0 - N_{\mathcal{Z}_0}) = \frac{N_0 - N_{\mathcal{Z}_0}}{N_0}.
        \end{align}
        \begin{align}
            E[\mathcal{Y}_v(\Influence_v \geq \hat{\theta}_v)]  & = E[\{ W_v(\mathbf{a_i}) \mid \Influence_v(\mathbf{a_i}) \geq \hat{\theta}_v\}] = \frac{1}{N_1} \cdot N_1 = 1.
        \end{align}
        \begin{align}
            E[\mathcal{Y}_v(\Influence_v \geq \hat{\theta_v}) & - \mathcal{Y}_v(\Influence_v < \hat{\theta_v})] = 1 - \frac{N_0 - N_{\mathcal{Z}_0}}{N_0} < 1.
        \end{align}
    \end{case}
    \begin{case}[$\hat{\theta}_v = \theta_v$]
		\begin{align}
            E[\mathcal{Y}_v(\Influence_v < \hat{\theta}_v)]  & = E[\{ W_v(\mathbf{a_i}) \mid \Influence_v(\mathbf{a_i}) < \hat{\theta}_v\}]
            \nonumber\\
			& = \frac{1}{N_0} \sum_i W_v(\mathbf{a_i}) = 0.
			\\
			E[\mathcal{Y}_v(\Influence_v \geq \hat{\theta}_v)]  & = E[\{ W_v(\mathbf{a_i}) \mid \Influence_v(\mathbf{a_i}) \geq \hat{\theta}_v\}]
            \nonumber \\
            & = \frac{1}{N_1} \sum_i W_v(\mathbf{a_i}) = 1.
			\\
			E[\mathcal{Y}_v(\Influence_v \geq \hat{\theta_v}) & - \mathcal{Y}_v(\Influence_v < \hat{\theta_v})] = 1 - 0 = 1.
        \end{align}
        The maximum causal effect only occurs when $\hat{\theta}_v = \theta_v$.
		Therefore if we can estimate the trigger that maximizes CAPE we find the true node threshold.
		As a result, we can use any trigger-based HTE estimation method to estimate node thresholds.
		\qedhere
    \end{case}
\end{proof}
}

\subsubsection{Trigger-based Causal Trees}

Adapting trigger-based causal tree methods to find individual thresholds is straightforward.
Causal trees work similarly to decision trees in that they greedily split using a partition function.
The main difference between causal trees and decision trees is that the goal of a causal tree is to estimate CATE for different populations of individuals rather than to predict a label.
Causal trees work by maximizing a partition function based on the difference of mean in outcomes while keeping low variance outcomes~\cite{athey2016recursive}.
The intuition is that finding splits that maximize the difference in means finds the heterogeneity in effect.

Tran and Zheleva developed a trigger-based causal tree method called CT-HV~\cite{tran-aaai2019}.
CT-HV works by introducing a validation set for generalizing CATE estimations and reducing variance in outcomes.
In order to learn triggers, an additional search is done at each tree split to find the trigger that maximizes the effect estimation for that split.
We can utilize CT-HV for threshold estimation through CAPE with a trigger.
Instead of an individual-level treatment, we use the activation influence, $\Influence^t$, to estimate triggers.
\yesappendix{
Details of their algorithm are in the Appendix.
}

\noappendix{
Here we describe the Causal Tree (CT-HV) algorithm proposed in~\cite{tran-aaai2019}. Let $\mathbf{X}^{\ell}$ be the features in partition $\ell$, which represents a node in the tree which contains $N_\ell$ samples, and let $\hat{\mu}_1(\ell)$ and $\hat{\mu}_0(\ell)$ be the mean of outcomes when treated and non-treated.
The estimate of CATE of any partition is $\hat{\tau_c}(\mathbf{X}^\ell) = \hat{\mu}_1(\ell) - \hat{\mu}_0(\ell)$.
Given a partition $\ell$ that needs to be partitioned further into two children $\ell_1, \ell_2$, the causal tree algorithm finds the split on features that maximizes the weighted CATE in each child:
\begin{equation}
    \maxl_{\ell_1, \ell_2} N_{\ell_1} \cdot \hat{\tau_c}(\mathbf{X}^{\ell_1}) +  N_{\ell_2} \cdot \hat{\tau_c}(\mathbf{X}^{\ell_2}).
\end{equation}
We refer to the two quantities $N_{\ell_1} \cdot \hat{\tau_c}(\mathbf{X}^{\ell_1})$ and $N_{\ell_2} \cdot \hat{\tau_c}(\mathbf{X}^{\ell_2})$ as partition measures for $\ell_1$ and $\ell_2$.
In our work, we adapt the CT-HV algorithm~\cite{tran-aaai2019} for the problem of node threshold estimation.
}

\noappendix{
In order to learn triggers, an additional search is done at each split to find the trigger that maximizes the effect estimation.
Let $F(\ell)$ represent the partition measure for CT-HV and let $\hat{\theta_v}$ be some estimated trigger for partition $\ell$.
Define $M_1(\ell, \hat{\theta_v})$ and $M_0(\ell, \hat{\theta_v})$ to be the expected mean outcomes when above and below the trigger $\hat{\theta_v}$, and $F(\ell, \hat{\theta_v})$ be the partition measure with trigger $\hat{\theta_v}$: $F(\ell, \hat{\theta_v}) = M_1(\ell, \hat{\theta_v}) - M_0(\ell, \hat{\theta_v}) = E[\mathcal{Y}_\ell(\Influence \geq \hat{\theta_v}) - \mathcal{Y}_\ell(\Influence < \hat{\theta_v})].$
Then to find a trigger in partition $\ell$, we find the trigger that maximizes the partition measure: \( F(\ell, \hat{\theta_v}) \), which results in a unique trigger in each partition.
}

\subsubsection{ST-Learner}

Now we present a novel causal \emph{meta-learner} for estimating triggers for heterogeneous effects.
A causal meta-learner uses the outputs of base learners for estimating effects.
A base learner can be any regression or classification method, such as Linear Regression or Decision Trees.
K{\"u}nzel et al.\ name two commonly used meta-learners in causal inference literature, S- and T-learners, and develop a meta-learner called X-Learner~\cite{kunzel-pnas19}.
However, none of them consider the problem of trigger-based HTE and are not directly suitable for solving the problem of threshold estimation.
Following the idea of causal meta-learners, we propose a novel algorithm, \emph{ST-Learner}\footnote{The ``S'' refers to a \emph{single} learner and ``T'' is for \emph{triggers}, following the naming scheme.}.

The goal of our ST-Learner is to learn the trigger that maximizes the effect.
A base learner is trained to predict the outcome $(Y_v)$ given all node features $(\mathbf{X}_v)$ and the activation influence $(\Influence_v)$: $E[Y_v | \mathbf{X}_v, \Influence_v]$.
Once a learner has been built on the training data (i.e., known activations), the next step is to estimate the trigger for each node.
The estimation considers different possible values of $\Influence_v$ and consists of two steps: outcome prediction and trigger estimation.

Let $\Alltreatments = \{\alltreatments_1, \alltreatments_2, \dots, \alltreatments_n\}$ be all the possible treatment values (i.e., activation influence levels $\Influence_v)$ in the data.
Define $\Theta = \{r_1, r_2, \dots, r_m\}$,  $m \leq n$, to be a set of triggers.
We use $r_i$ to refer to any potential trigger, while $\theta_v$ is a \emph{node's} individual trigger.
In practice, we can consider all potential values of treatment in the training data to be possible triggers, or we can consider a discretization of these values ($m \leq n$).
For any node, we can compute $n$ predictions for all possible treatment values $\alltreatments_k \in \Alltreatments$: ${E}[Y_v | \mathbf{X}_v, \Influence=\alltreatments_k]$.

With the predictions based on different activation influence levels, we can estimate CAPE with a trigger.
Let $\Theta_i^1 = \{\alltreatments_k : \alltreatments_k \geq r_i \}$ and $\Theta_i^0 = \{\alltreatments_j : \alltreatments_j < r_i\}$ be the set of treatment values above and below some trigger $r_i$.
For an arbitrary trigger $r_i$, we get the trigger estimation for ST-Learner:
\begin{align}
    \nonumber
    \argmax_{r_i} \hat{\mathcal{P}}_s(\mathbf{x})
    & =
    \frac{1}{|\Theta_i^1|} \sum_{\alltreatments_k \geq r_i} {E}[Y_v | \mathbf{X}_v=\mathbf{x}, \Influence=\alltreatments_k]
    \\ & -
    \frac{1}{|\Theta_i^0|} \sum_{\alltreatments_j < r_i} {E}[Y_v | \mathbf{X}_v=\mathbf{x}, \Influence=\alltreatments_j]. \nonumber
    \label{eq:st-learner}
\end{align}
We use the base learner to estimate the outcomes above and below the trigger by taking average outcomes.
\aaaipre{
This results in a time complexity of $O(|B| \cdot L)$ where $O(L)$ is the complexity of prediction for a base learner.
}
\yesappendix{
Pseudocode for ST-Learner is available in the Appendix.
}

\noappendix{
\begin{algorithm}[t]
\small
\caption{ST-Learner}\label{alg:stlearner}
  \begin{algorithmic}[1]
    \Require Training set $(\mathbf{X}, Y, \Influence)$, test example $\mathbf{x}^{te}$, base learner $f$, possible treatments $\Alltreatments$, possible triggers $\Theta$
    \Ensure The causal effect estimate for $\mathbf{x}^{te}$
    
    \Function{Train}{$\mathbf{X}, Y, \Influence$}
        \State Fit $f$ to predict $Y$, $f$: $(\mathbf{X}, I) \rightarrow Y$
        \State $\implies f(\mathbf{X}, \Influence) = E[Y \mid \mathbf{X}, \Influence]$
    \EndFunction
    
    \Function{Predict}{$\mathbf{x}^{te}$}
        \For{each $\alltreatments_i$ in $\Alltreatments$}
            \State Compute and store $f(\mathbf{x}^{te}, \alltreatments_i)$
        \EndFor
        \State max\_trigger $= 0$, max\_effect $= 0$
        \For{each $r_i$ in $\Theta$}
            \State Compute $t_1 = \frac{1}{|\Theta_i^1|} \sum_{\alltreatments_k \geq r_i} f(\mathbf{x}^{te}, \alltreatments_k)$
            \State Compute $t_0 = \frac{1}{|\Theta_i^0|} \sum_{\alltreatments_j < r_i} f(\mathbf{x}^{te}, \alltreatments_j)$
            \State $e = t_1 - t_0$
            \If{$e >$ max\_effect}
                \State max\_trigger $= r_i$, max\_effect $= e$
            \EndIf
        \EndFor
        \State \textbf{return} max\_trigger, max\_effect
    \EndFunction
  \end{algorithmic}
\end{algorithm}

Pseudocode is provided for ST-Learner in Algorithm~\ref{alg:stlearner}.
We show the training and prediction subroutines.
In the train subroutine, we train a base learner, \( f \) to predict the outcome \( Y \) using node features \( \mathbf{X} \) and the activation influence \( I \).
To predict a trigger for a new test example, we compute and score the predicted outcome for all treatment values found in the dataset.
Then, we find the trigger that maximizes the average difference above and below that trigger.
}

\section{Experiments}\label{sec:experiments}

We study the effectiveness of our proposed methods
using four synthetic network generation models and three real-world datasets.
We consider three tasks: threshold prediction, activated node prediction, and diffusion size prediction.

\subsection{Experimental setup}\label{sec:setup}

Our main goal in this work is to estimate node thresholds (task 1) and understand their role in downstream tasks, such as activated node prediction (task 2) and diffusion size prediction (task 3).
Comparing LTM to other diffusion models (e.g., cascade models) is outside the scope of this work.
To compare the effectiveness of threshold prediction on these three tasks, we estimate the node thresholds given the data at a snapshot (e.g., end of Jan 2016).
A snapshot is the current structure and activations at network time $t$.
Each model predicts thresholds based on individual snapshots of data.
For example, models learn thresholds using the network snapshot at the end of Jan 2016 ($\mathcal{D}_0$) and estimate node thresholds for all inactivated nodes in the test data.
In this way, we also assess how having snapshots with more activations, and thus more training data, impacts predictions.
\noappendix{
\begin{figure}[t]
	\centering
	\includegraphics[width=0.9\columnwidth]{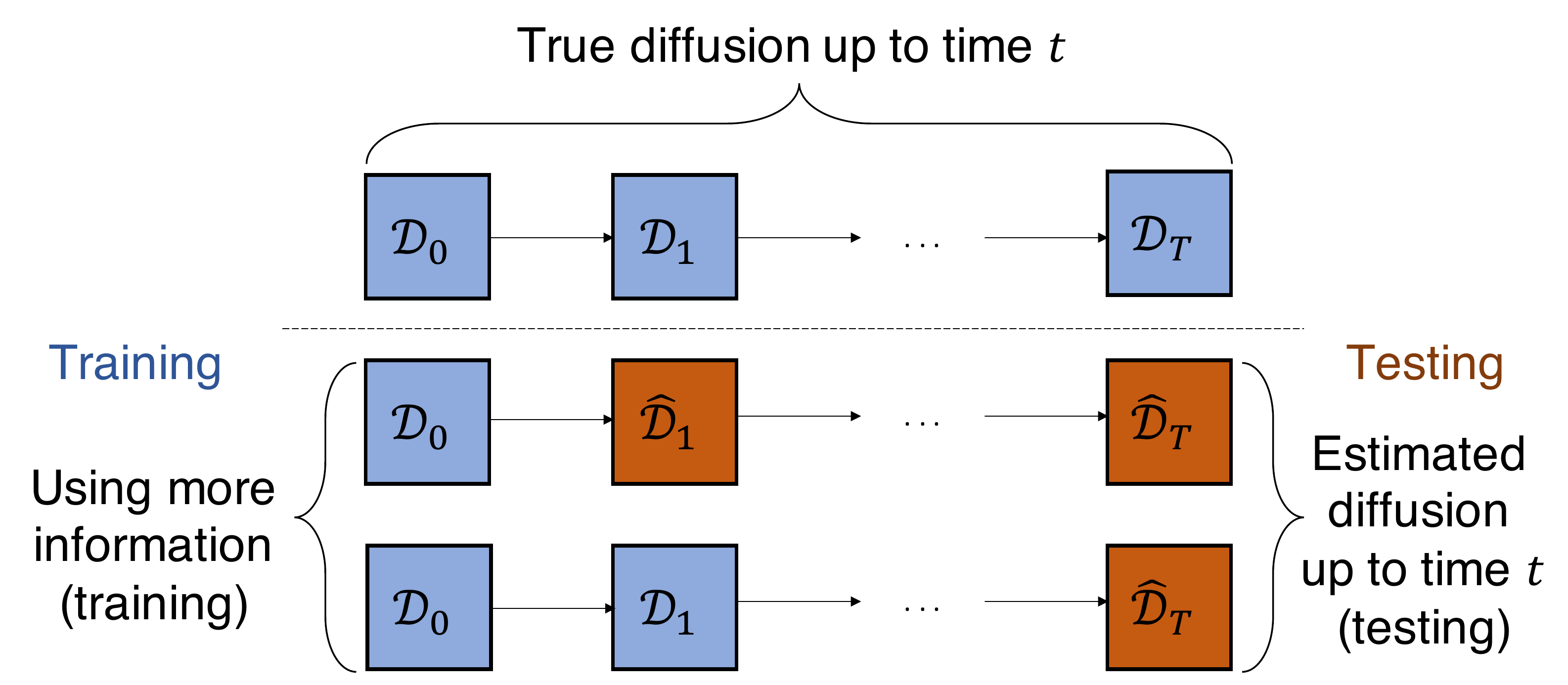}
	\caption{Representation of our experimental setup. Given the true activations up to time $T$ (first row), we train on different sets. $\mathcal{D}_i$ represents the set of all activated nodes up to time $i$ and $\hat{\mathcal{D}}_i$ represents the predicted activations. 
	}\label{fig:experiment}
	\vspace{-1.0em}
\end{figure}
Figure~\ref{fig:experiment} shows the general experimental setup.
}

\yesappendix{
\begin{figure*}[!ht]
	\centering
    \begin{subfigure}{\legendsize}
        \includegraphics[width=\textwidth]{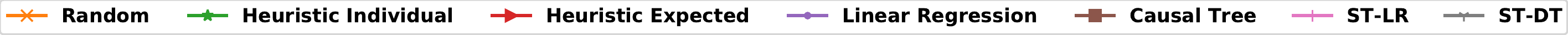}    
    \end{subfigure}

	\begin{subfigure}[b]{\subfiguresize}
		\includegraphics[width=\textwidth]{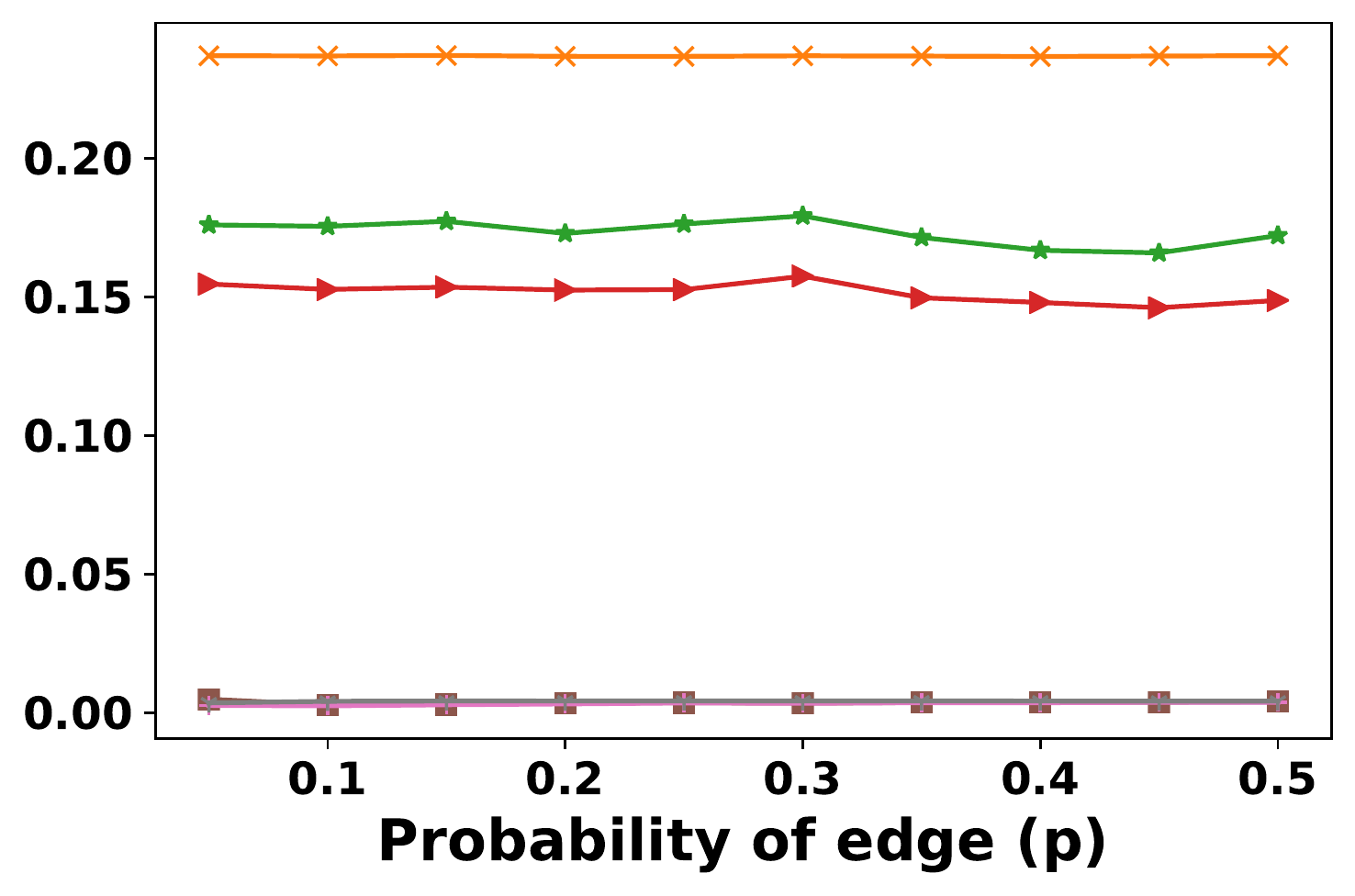}
		\caption{Erdos-Renyi}\label{fig:linear-erdos}
	\end{subfigure}
	\begin{subfigure}[b]{\subfiguresize}
		\includegraphics[width=\textwidth]{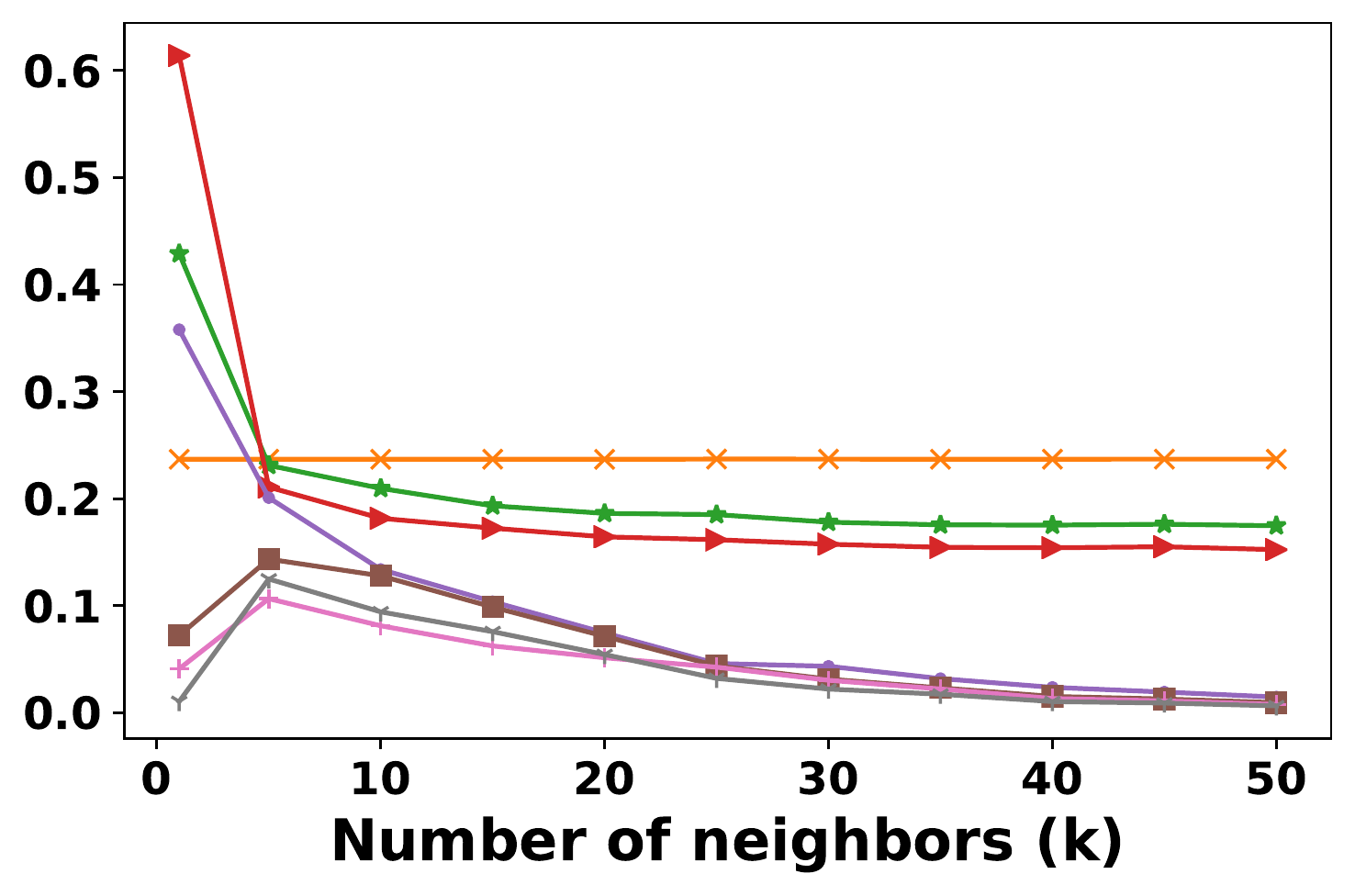}
		\caption{Pref. attachment}\label{fig:linear-pref}
	\end{subfigure}
	\begin{subfigure}[b]{\subfiguresize}
		\includegraphics[width=\textwidth]{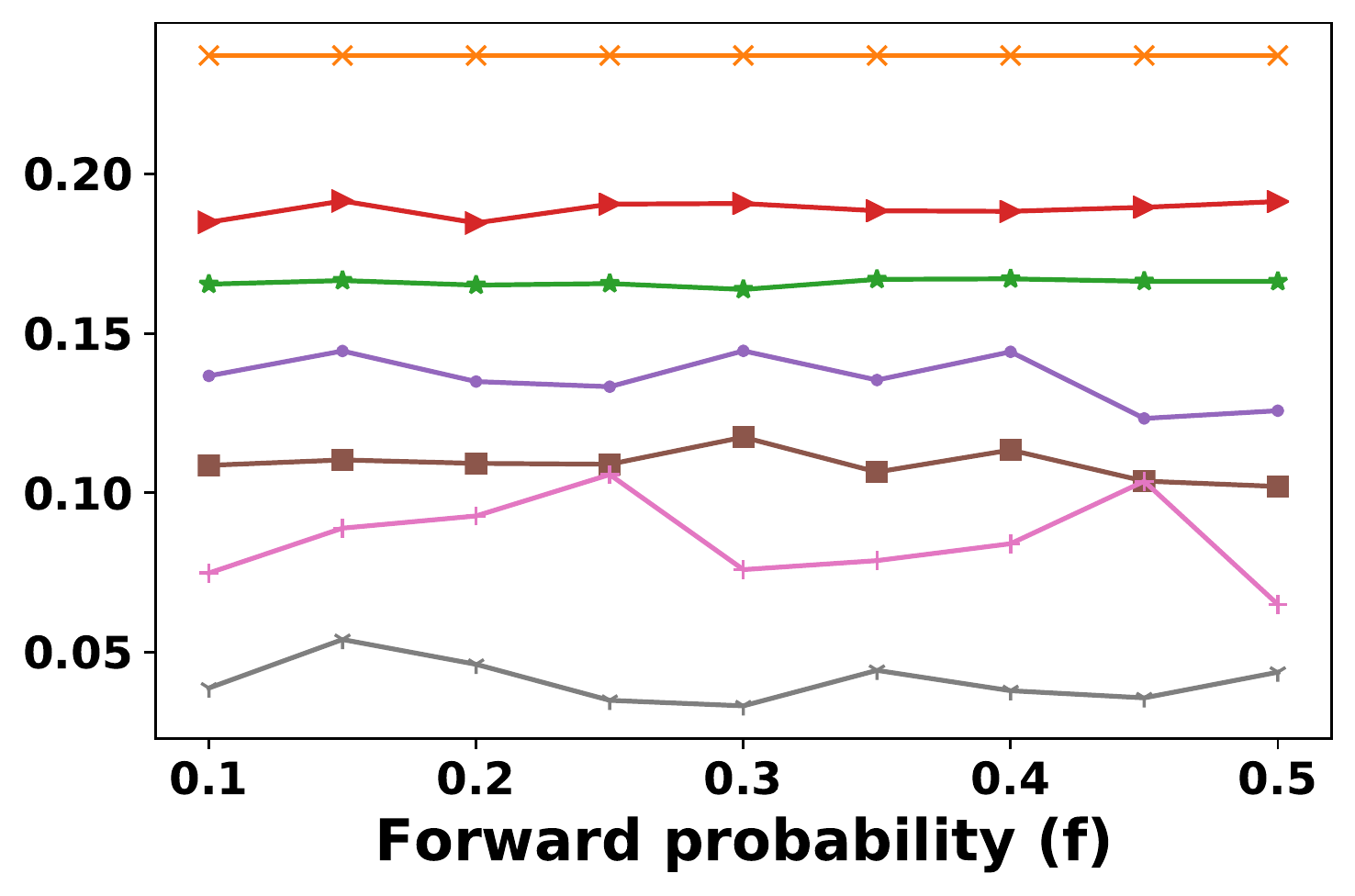}
		\caption{Forest fire}\label{fig:linear-forest}
	\end{subfigure}
	\begin{subfigure}[b]{\subfiguresize}
		\includegraphics[width=\textwidth]{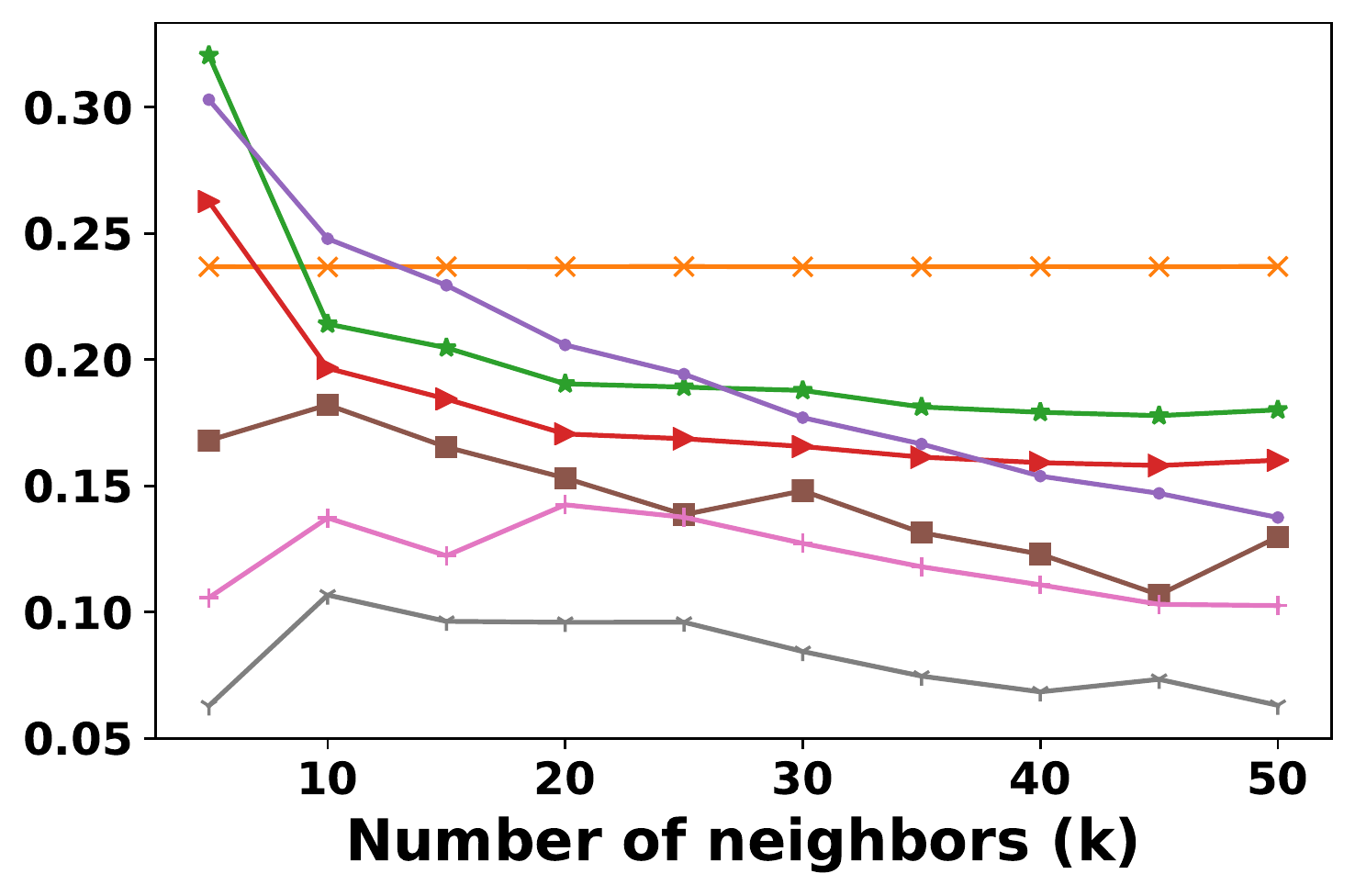}
		\caption{Watts-Strogatz}\label{fig:linear-watts}
	\end{subfigure}

	\caption{\kdd{MSE of threshold prediction for the \emph{Linear} setups over different parameters.}}\label{fig:synthetic_mse}
	\vspace{\myspace}
\end{figure*}
}
\noappendix{
\begin{figure*}[!ht]
	\centering
    \begin{subfigure}{\legendsize}
        \includegraphics[width=\textwidth]{figures/legend_cut.pdf}    
    \end{subfigure}

	\begin{subfigure}[b]{\subfiguresize}
		\includegraphics[width=\textwidth]{figures/synthetic/linear/erdos-renyi.pdf}
		\caption{Erdos-Renyi}\label{fig:linear-erdos}
	\end{subfigure}
	\begin{subfigure}[b]{\subfiguresize}
		\includegraphics[width=\textwidth]{figures/synthetic/linear/pref-attach.pdf}
		\caption{Pref. attachment}\label{fig:linear-pref}
	\end{subfigure}
	\begin{subfigure}[b]{\subfiguresize}
		\includegraphics[width=\textwidth]{figures/synthetic/linear/forest-fire.pdf}
		\caption{Forest fire}\label{fig:linear-forest}
	\end{subfigure}
	\begin{subfigure}[b]{\subfiguresize}
		\includegraphics[width=\textwidth]{figures/synthetic/linear/small-world.pdf}
		\caption{Watts-Strogatz}\label{fig:linear-watts}
	\end{subfigure}

	\begin{subfigure}[b]{\subfiguresize}
		\includegraphics[width=\textwidth]{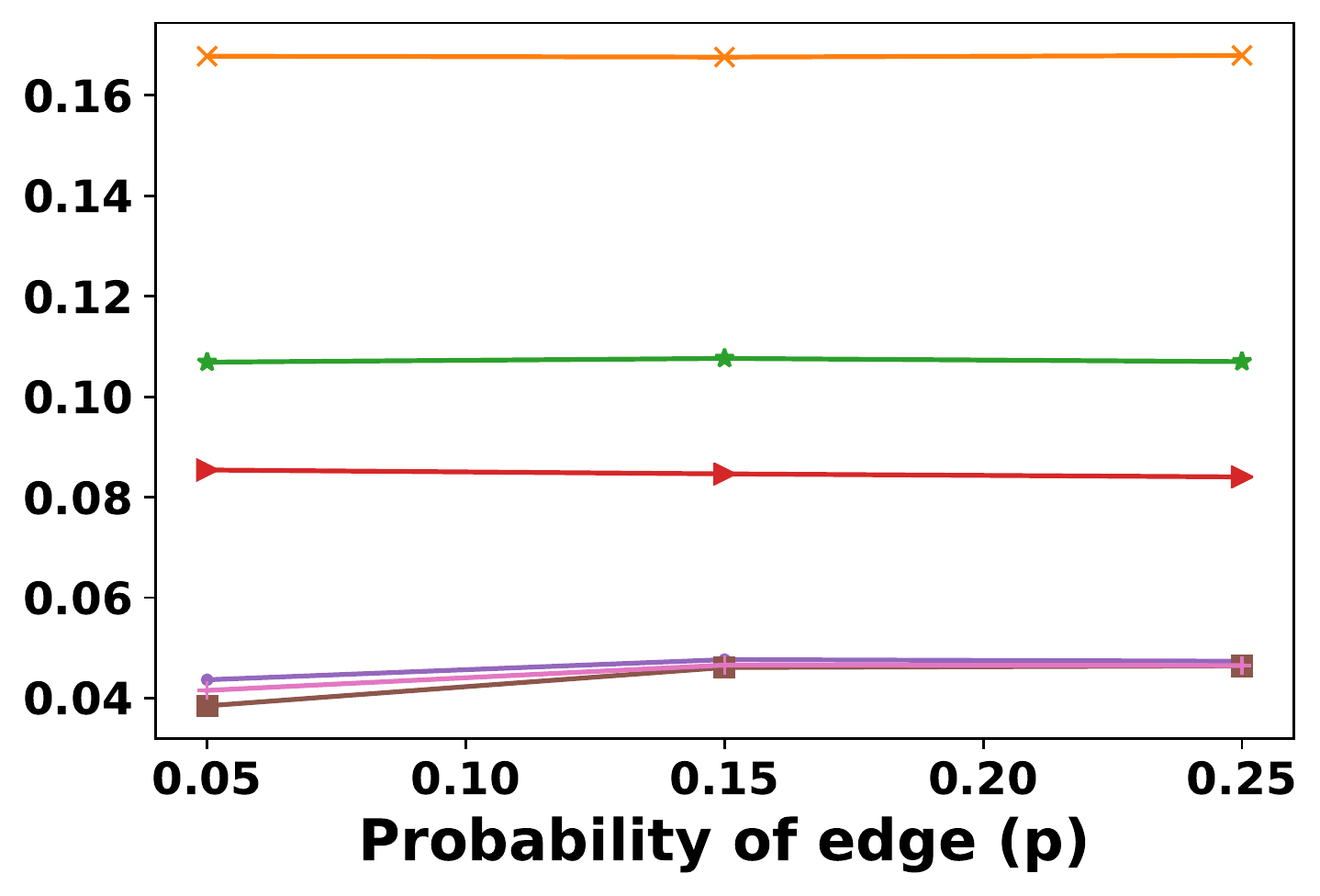}
		\caption{Erdos-Renyi}\label{fig:quad-erdos}
	\end{subfigure}
	\begin{subfigure}[b]{\subfiguresize}
		\includegraphics[width=\textwidth]{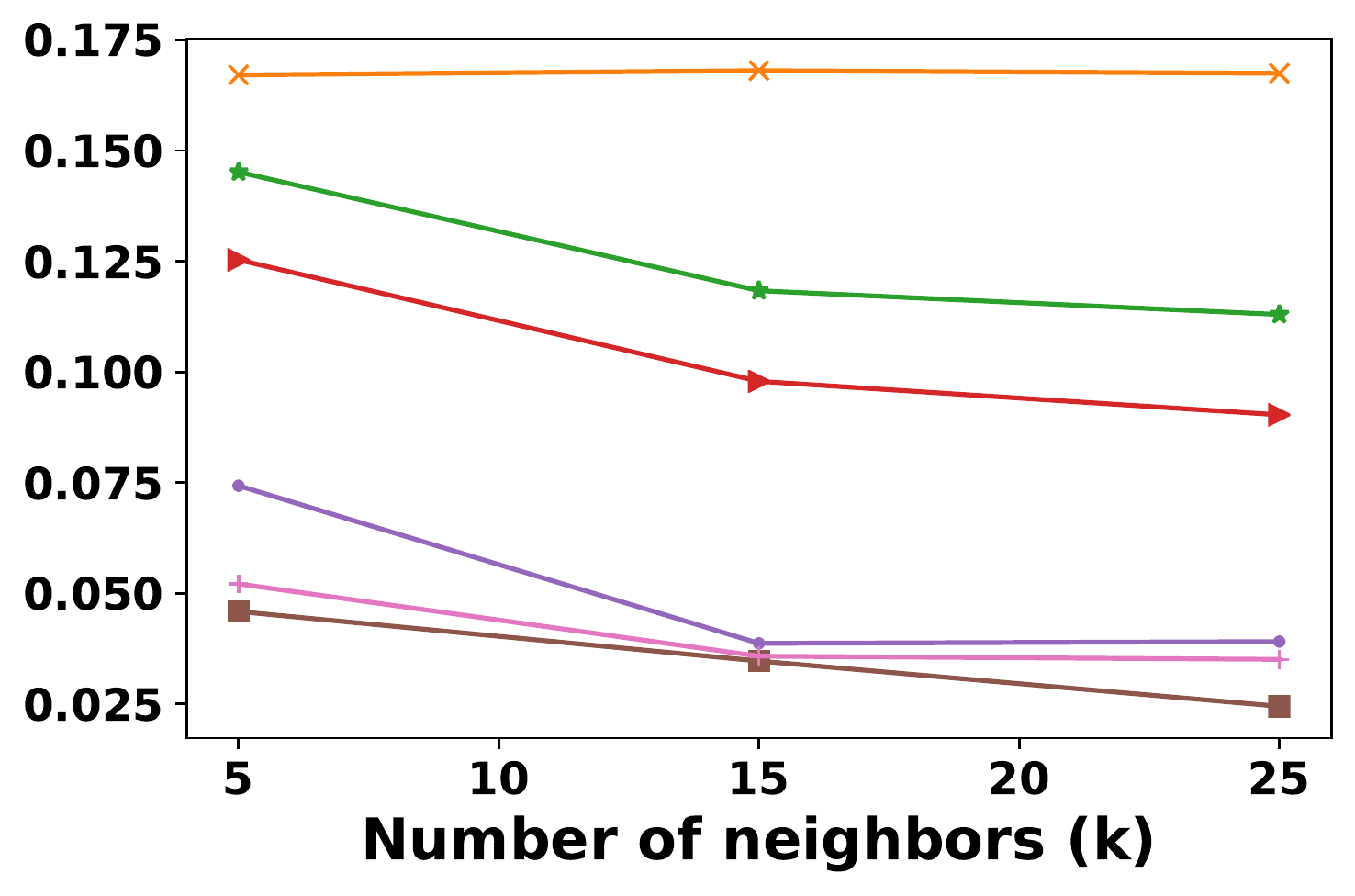}
		\caption{Pref. attachment}\label{fig:quad-pref}
	\end{subfigure}
	\begin{subfigure}[b]{\subfiguresize}
		\includegraphics[width=\textwidth]{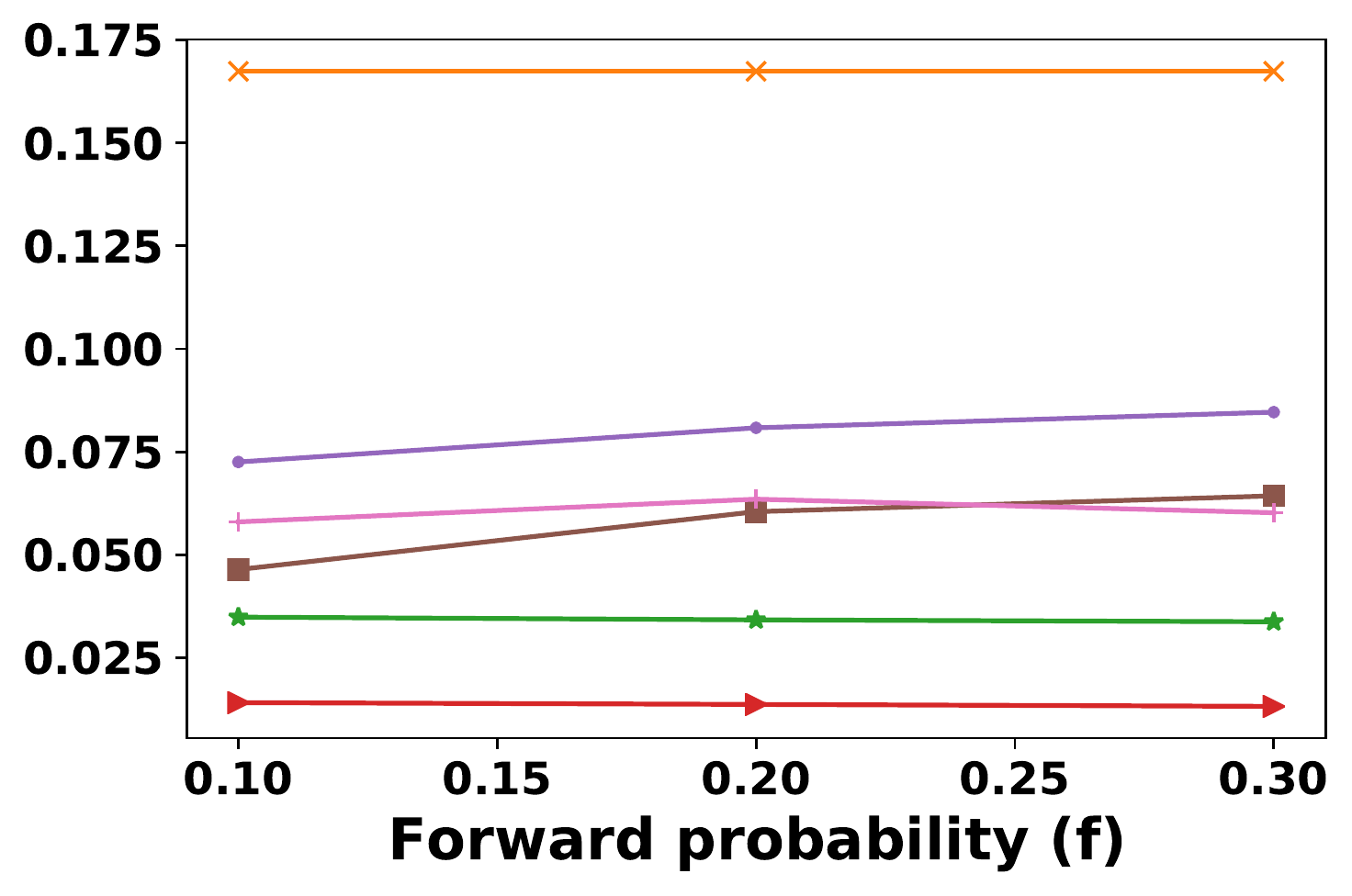}
		\caption{Forest fire}\label{fig:quad-forest}
	\end{subfigure}
	\begin{subfigure}[b]{\subfiguresize}
		\includegraphics[width=\textwidth]{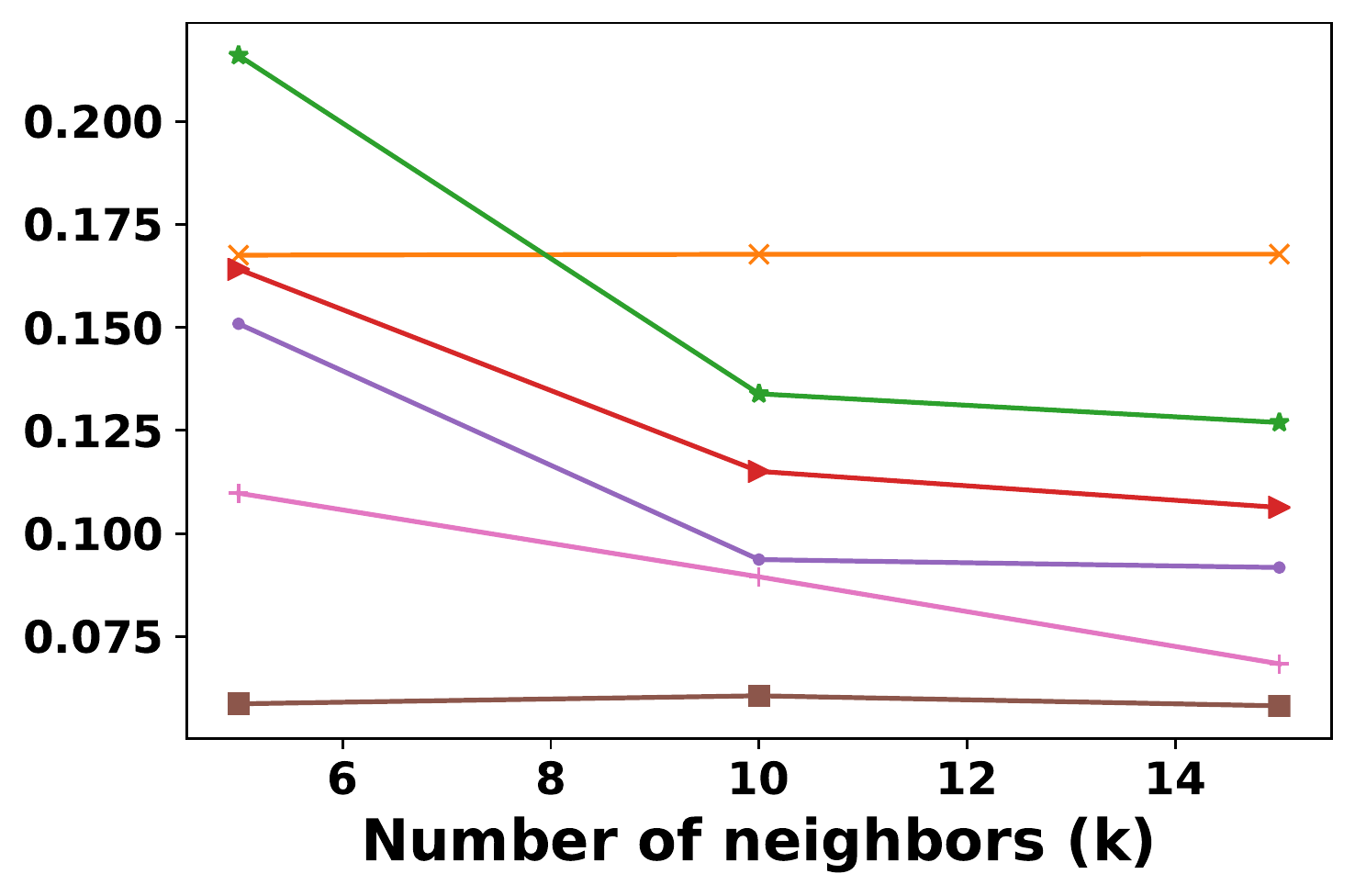}
		\caption{Watts-Strogatz}\label{fig:quad-watts}
	\end{subfigure}
	\caption{\kdd{MSE of threshold prediction for the \emph{Linear} (top row) and \emph{Quadrant} (bottom row) setups over different parameters.}}\label{fig:synthetic_mse_appendix}
\end{figure*}
}

We use four baseline threshold estimation methods, two of which are based on a recent LTM survey~\cite{talukder2019threshold}.
\emph{Heuristic Expected} computes the same threshold for all nodes in the network.
\emph{Heuristic Individual} estimates a range of values and samples individual node thresholds randomly from that range.
We also employ a baseline called \emph{Random}, that assign node thresholds uniformly random from 0 to 1.
In addition, we use a more practical baseline for threshold estimation using \emph{Linear Regression}.
To get labels for Linear Regression, we take all activated nodes in the network and estimate their threshold by computing the proportion of activated neighbors.
All baseline methods focus on node threshold prediction for LTM\@.

We compare \emph{Causal Tree} and \emph{ST-Learner} to the baseline methods.
\aaaipre{
For \emph{ST-Learner}, we use two base learners, Linear Regression (ST-LR) and Decision Trees (ST-DT).
}
For neighbor influences, we use degree centrality: $w_{uv} = 1/|N(v)|$, where $|N(v)|$ is the degree of node $v$.

\subsubsection{Task 1: node threshold prediction}
For the first task, the goal is to accurately predict individual node thresholds.
Given a snapshot of the data at time $t$, we estimate a threshold for each node.
To evaluate the effectiveness of each model, we use the average mean squared error of predicted thresholds across all snapshots.
\begin{equation}
	\text{MSE} = \frac{1}{T \cdot|V|} \sum_{t=1}^T \sum_{v}^{|V|} \Big(\theta_v - \hat{\theta}_v^{(t)} \Big)^2.
\end{equation}
This task is evaluated based on synthetic data because the true thresholds are unknown in real-world data.

\subsubsection{Task 2: activated node prediction}
In the second task, we are interested in how well each model can predict which specific nodes activate at each time step.
To do this, we use the predicted node thresholds and simulate a diffusion according to LTM\@.
Let $\mathcal{D}_t$ be the set of activated nodes at time $t$ and $\hat{\mathcal{D}}_t$ be the set of predicted activated nodes at time $t$.
We compute the average Jaccard index on all sets as:
\begin{equation}
    J = \frac{1}{T} \sum_{t=0}^T \frac{|\mathcal{D}_t \cap \hat{\mathcal{D}}_t|}{|\mathcal{D}_t \cup \hat{\mathcal{D}}_t|}.
\end{equation}
The Jaccard index shows how accurate a model's estimated thresholds are for predicting activated nodes.

\subsubsection{Task 3: diffusion size prediction}
The final task is diffusion size prediction.
Using the simulated diffusion, we plot the number of predicted activations and compare them to the true number of activations at each time step.
This task shows how node threshold prediction impacts diffusion size prediction.
We refer to the diffusion size as the \emph{reach}.

\subsection{Datasets}\label{sec:datasets}

\begin{table*}[!ht]
	\centering
	\caption{Our models predict the specific nodes that will activate most accurately based on highest average Jaccard index.}\label{tab:acc}
	\resizebox{0.85\textwidth}{!}{
		\begin{tabular}{P{3.0cm} *{13}{M{1.15cm}}}
			\toprule
			& \multicolumn{4}{c}{Linear setup}
			& \multicolumn{4}{c}{Quadrant setup}
			& \multicolumn{4}{c}{Real-world data}
			\\
			\cmidrule(lr){2-5}
			\cmidrule(lr){6-9}
			\cmidrule(lr){10-13}
			Threshold prediction method 
			& Erdos-Renyi & Pref. Attach. & Forest Fire & Watts-Strogatz %
			& Erdos-Renyi & Pref. Attach. & Forest Fire & Watts-Strogatz %
			& Hateful 2016 & Hateful 2017 & Cannabis & Higgs %
			\\
			\midrule
			Random 
			& 0.6981 & 0.7986 & 0.9094 & 0.8082 %
			& 0.9830 & 0.8319 & 0.9187 & 0.7877 %
			& 0.2450 & 0.2262 & 0.2033 & 0.6287 %
			\\
			
			Heuristic Expected 
			& 0.6268 & 0.7415 & 0.9649 & 0.7557 %
			& 0.9783 & 0.7736 & 0.9630 & 0.7082 %
			& 0.2608 & 0.2300 & 0.2433 & 0.6429 %
			\\
			
			Heuristic Individual 
			& 0.9127 & 0.8728 & 0.9103 & 0.8398 %
			& 0.9785 & 0.7801 & 0.9692 & 0.7196 %
			& 0.2617 & 0.2301 & 0.2691 & 0.6436 %
			\\
			
			Linear Regression 
			& 0.9082 & 0.8589 & 0.9454 & 0.8306 %
			& 0.9891 & 0.8532 & 0.9494 & 0.7940 %
			& 0.1522 & 0.2297 & 0.3674 & 0.6449 %
			\\
			
			\midrule

			Causal Tree 
			& 0.9599 & 0.9357 & 0.9509 & 0.8383 %
			& 0.9978 & 0.9559 & 0.9638 & 0.8542 %
			& \bf0.7207 & 0.5793 & \bf0.7830 & 0.9427 %
			\\
			
			ST-Learner 
			& 0.9658 & 0.9471 & 0.9622 & 0.8618 %
			& 0.9972 & 0.9508 & 0.9526 & 0.8059 %
			& 0.6877 & 0.5785 & 0.7564 & 0.9457 %
			\\

			ST-Learner (Tree)
			& \bftab 0.9813 & \bftab 0.9777 & \bftab 0.9636 & \bftab 0.8669 %
			& \bftab 0.9983 & \bftab 0.9801 & \bftab 0.9663 & \bftab 0.8672 %
			& 0.7128 & \bftab 0.5988 & 0.7811 & \bftab 0.9734 %
			\\
			\bottomrule
			\\
		\end{tabular}
	}
	\vspace{\myspace}
\end{table*}

We evaluate task 1 under multiple synthetic network scenarios and tasks 2 and 3 with three real-world datasets.

\textbf{Synthetic datasets.}
We study four graph generation models: Erdos-Renyi~\cite{erdos-graph60}, preferential attachment~\cite{barabasi-science99}, forest fire~\cite{leskovec-kdd05}, and  Watts-Strogatz~\cite{watts-nature98} and two threshold generation models. 
For each set of \{network generation model, network parameter, threshold generation model\}, we run 10 simulations and report average results.
We set the number of nodes to 1000, and for each node, we randomly generate 100 node attributes from a Gaussian, $N(0, 1)$.

In Erdos-Renyi, we vary the probability of edge creation $p$ from $0.05$ to $0.5$.
In preferential attachment, the number of new attachments $k$ from $1$ to $50$.
For forest fire networks, we fix the backward probability of an edge to $0.1$ and vary the forward probability of an edge $f$ from $0.05$ to $0.5$.
For the Watts-Strogatz networks, we fix the probability of rewiring an edge to $0.1$, and vary nearest neighbors $k$ from $1$ to $50$.

In the first threshold generation, we use a random linear regression model with $10$ of $100$ attributes and normalize the output to be the user threshold, called the \emph{Linear} setup.
In the second threshold generation, we use $2$ features to separate thresholds into four quadrants, where each quadrant receives a single threshold uniformly from $U(0, 1)$, called the \emph{Quadrant} setup.
50 nodes are randomly activated, and diffusion events are generated based on LTM for $8$ time steps.

\textbf{Hateful Users.}
The Hateful Users dataset is a retweet network from Twitter, with 200 most recent tweets for each user~\cite{ribeiro2018characterizing}.
Each user is represented by the average \emph{Empath} category based on their tweets~\cite{fast2016empath}.
Empath captures a wide variety of topics such as violence, fear, and warmth.
A sample of users were selected to be annotated as \emph{hateful} or \emph{not hateful} and the rest were predicted based on the history of tweets using the methodology in~\cite{ribeiro2018characterizing}.

We estimate node thresholds for how ``hatefulness'' spreads through the network, where being activated means you change from \emph{not hateful} to \emph{hateful}.
We consider a month to month diffusion: how ``hatefulness'' diffuses on a month to month basis.
We look at two time periods: Jan 2016 to Dec 2016 and Jan 2017 to Oct 2017.

\textbf{Cannabis.}
The Cannabis dataset is a follower network.
The dataset covers all users who tweet about both cannabis and the e-cigarette Juul.
From the users who tweet about Juul, we identify those users who also tweet about cannabis or marijuana.
Empath categories are used as attributes.
We estimate thresholds for how cannabis tweets spread.
Activation means an individual tweets a cannabis related tweet.
We consider the period between Jan 2017 to Dec 2017.

\textbf{Higgs Boson.}
This dataset is based on the announcement of the Higgs-boson like particle at CERN on July 4, 2012.
The dataset was collected between July 1 and July 8 of 2012 and is a follower network on Twitter~\cite{de2013anatomy}.
We estimate node thresholds for the mention of the Higgs-boson discovery.
Activation means a user tweets about the Higgs-boson discovery.
There are no node attributes, so we construct features based on the graph.
We use degree centrality, both in-degree and out-degree, and counts of user and neighborhood tweets.
We consider hourly diffusions from July 4, 12:00am until July 8, 12:00am.

\yesappendix{
\begin{figure*}[!ht]
        \centering
        \begin{subfigure}{\legendsize}
            \includegraphics[width=\textwidth]{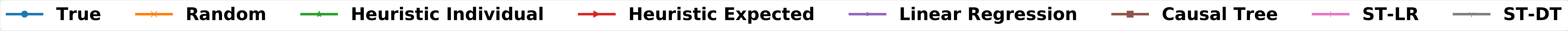}  
        \end{subfigure}
    
        \begin{subfigure}[t]{0.22\linewidth}
            \includegraphics[width=\textwidth]{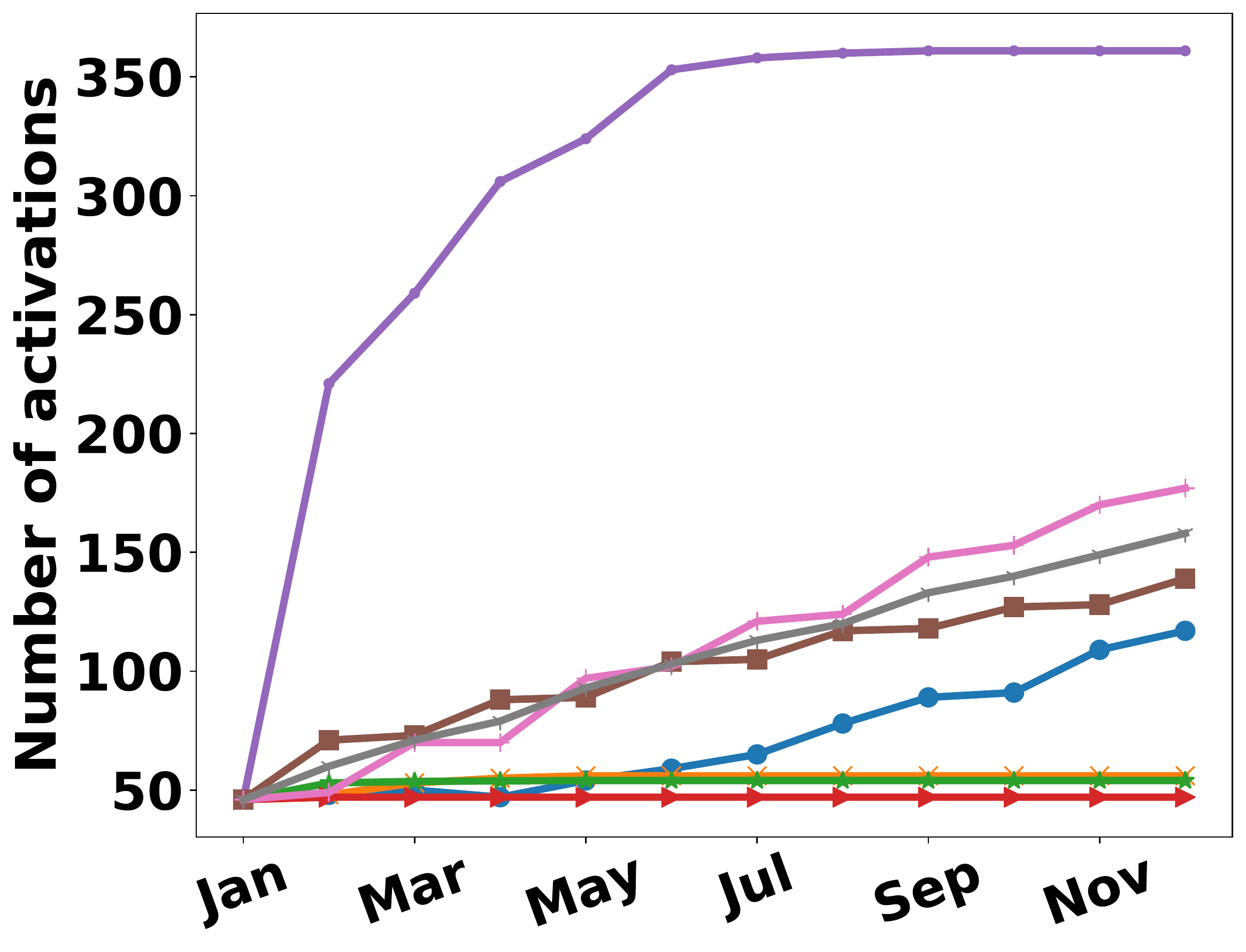}
            \caption{Hateful Users}\label{fig:hateful_main}
        \end{subfigure}
        \begin{subfigure}[t]{0.22\linewidth}
            \includegraphics[width=\textwidth]{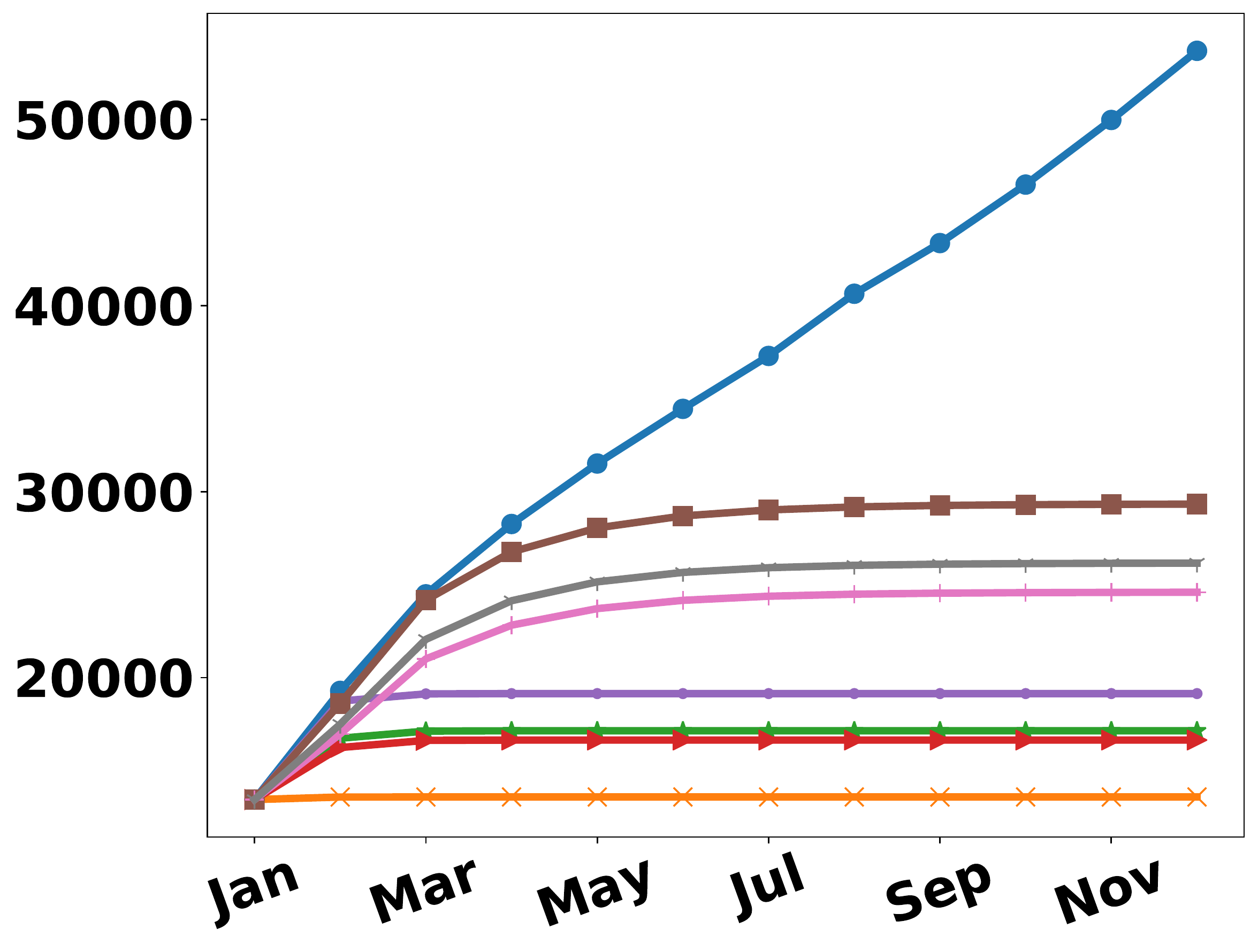}
            \caption{Cannabis}\label{fig:cannabis_main}
        \end{subfigure}
        \begin{subfigure}[t]{0.24\linewidth}
            \includegraphics[width=\textwidth]{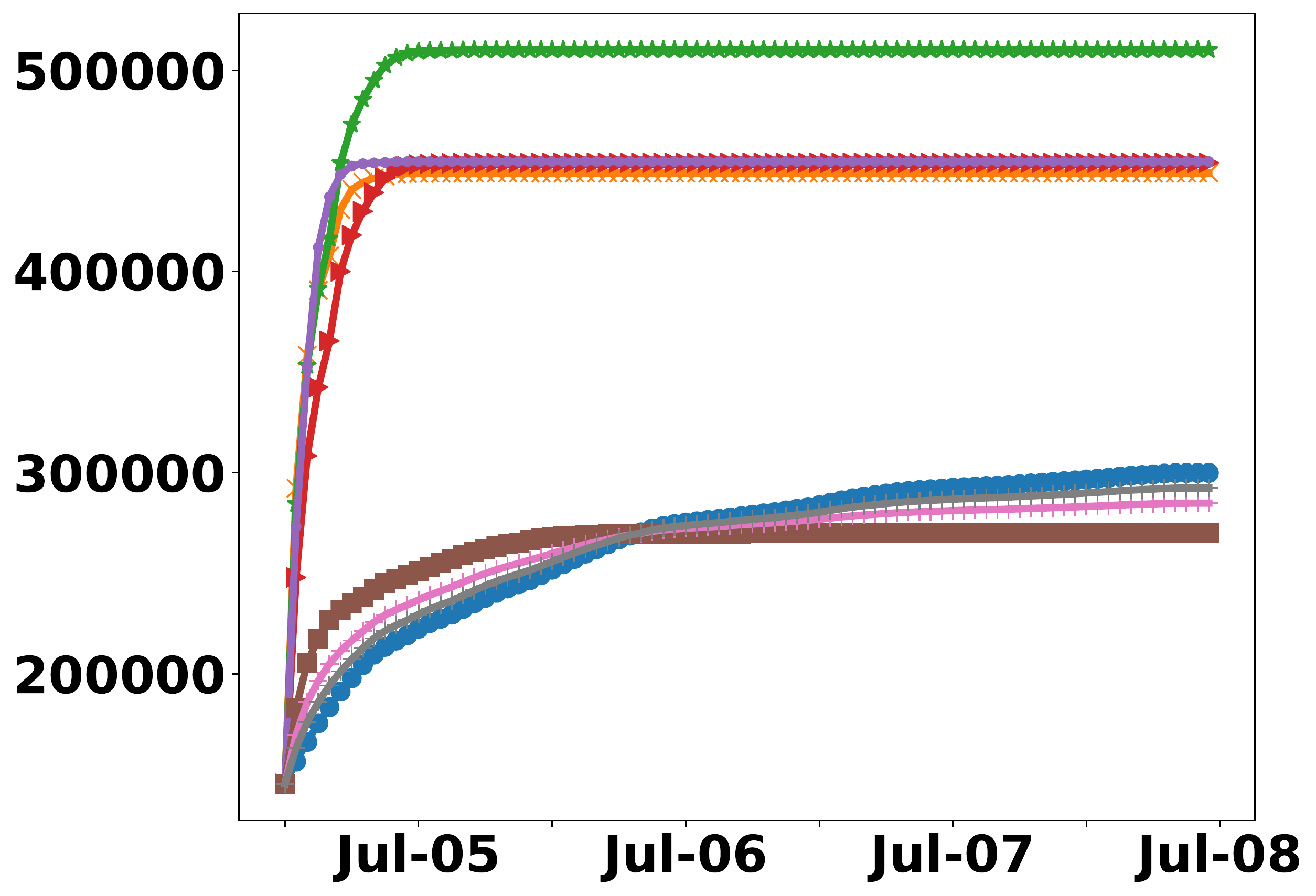}
            \caption{Higgs}\label{fig:higgs_main}
        \end{subfigure}
        \caption{
        Comparison of diffusion size prediction on three real world datasets. 
        Our models have the closest estimation of reach over longer time periods whereas the baselines incorrectly predict diffusion saturation in the early stages. 
        }\label{fig:real_compare}
        \vspace{\myspace}
    \end{figure*}

}
\noappendix{

\begin{figure*}[!ht]
    \centering
    \begin{subfigure}{\legendsize}
        \includegraphics[width=\textwidth]{figures/legend.pdf}  
    \end{subfigure}

    \begin{subfigure}[t]{\subfigurefirst}
        \includegraphics[width=\textwidth]{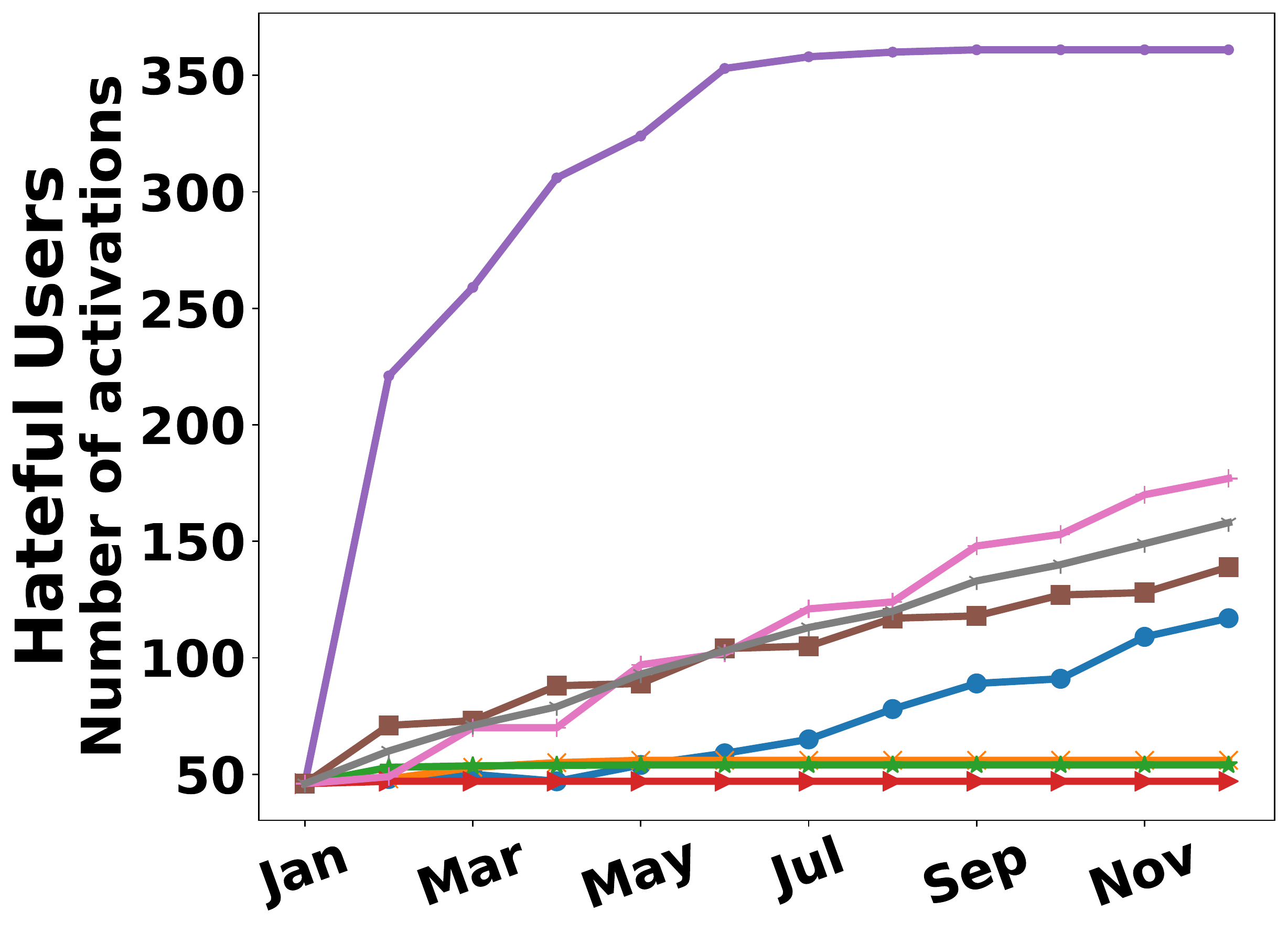}
        \caption{}
        \label{fig:hateful_compare_2016-0}
    \end{subfigure}
    \begin{subfigure}[t]{\subfiguresize}
        \includegraphics[width=\textwidth]{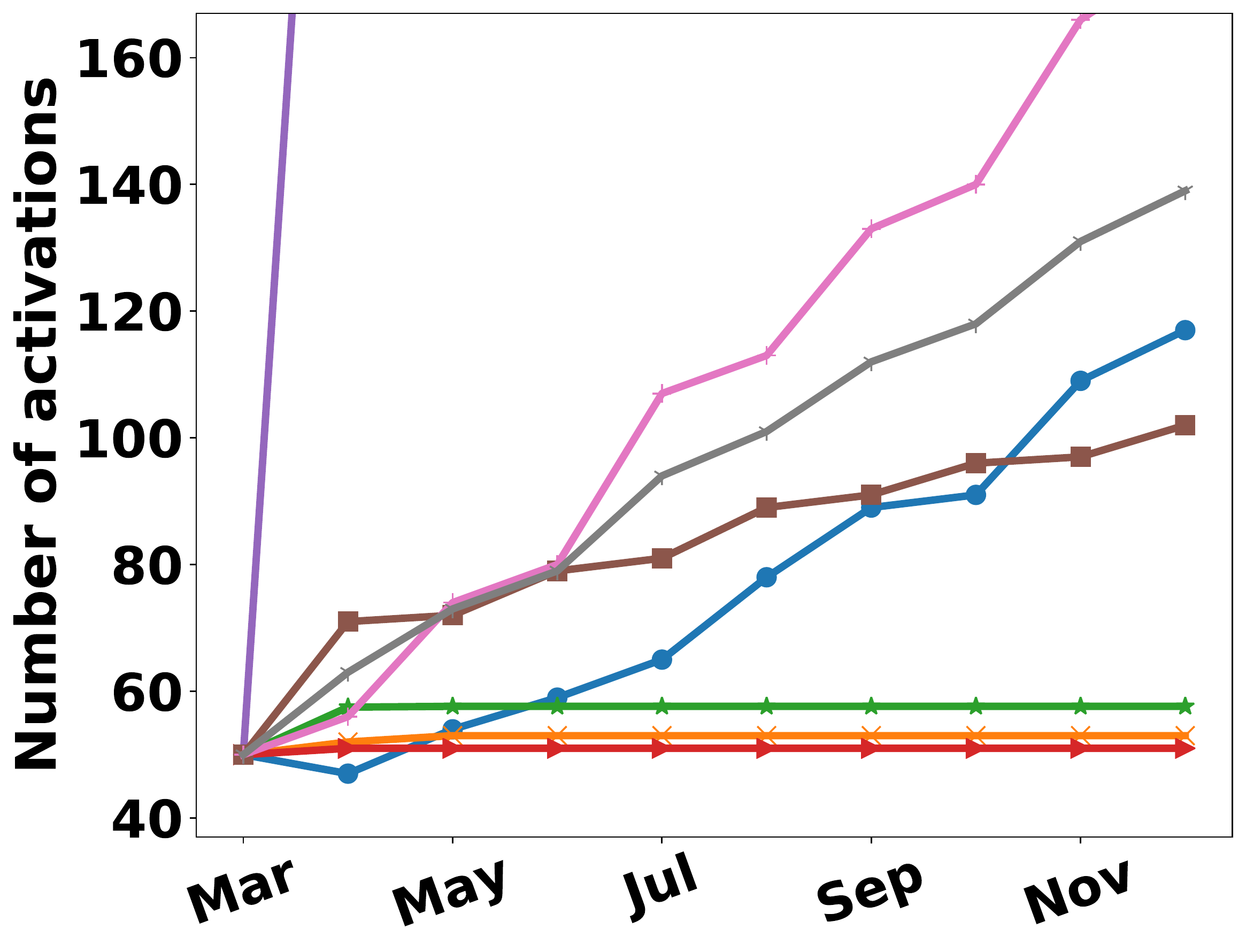}
        \caption{}
        \label{fig:hateful_compare_2016-2}
    \end{subfigure}
    \begin{subfigure}[t]{\subfiguresize}
        \includegraphics[width=\textwidth]{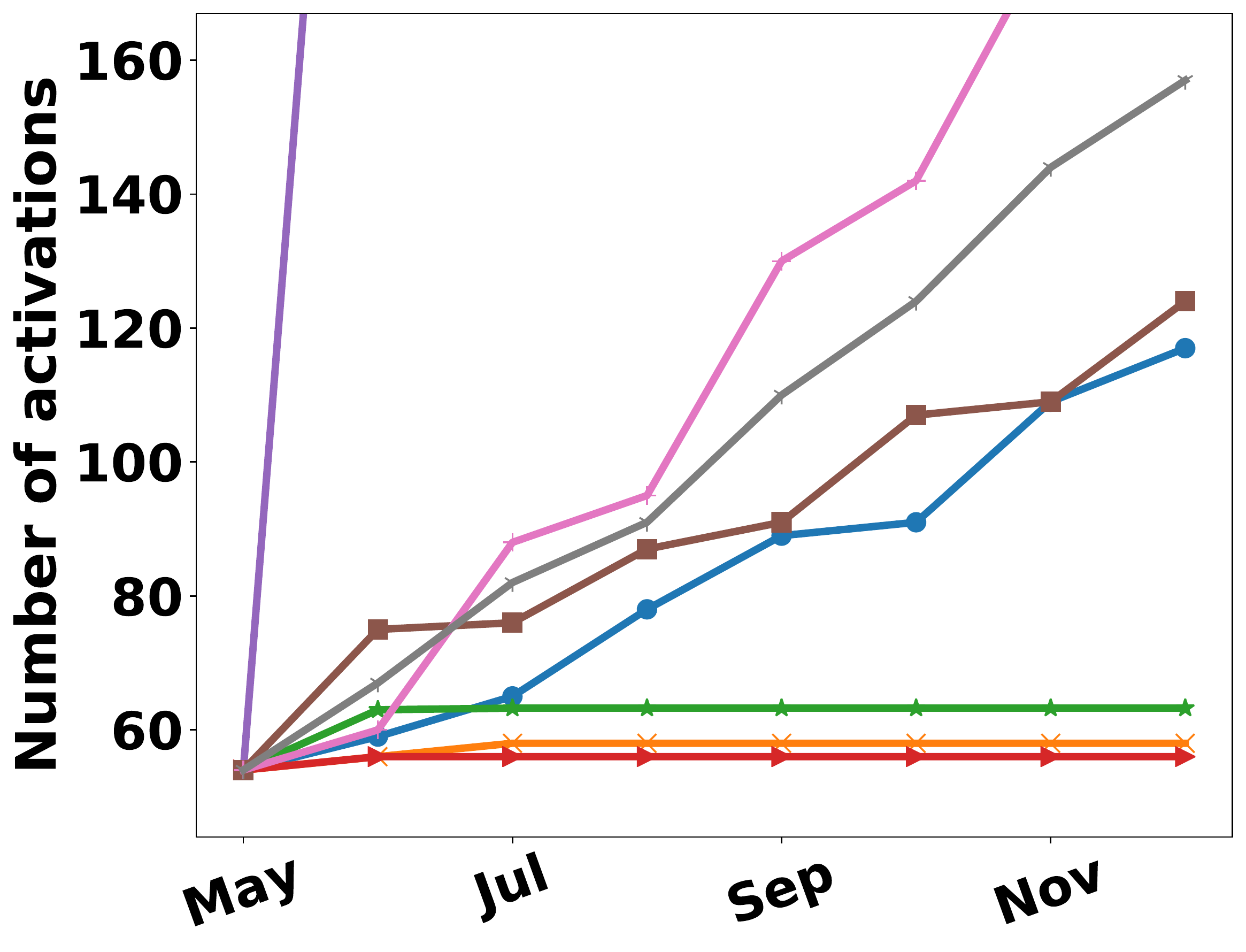}
        \caption{}
        \label{fig:hateful_compare_2016-4}
    \end{subfigure}
    \begin{subfigure}[t]{\subfiguresize}
        \includegraphics[width=\textwidth]{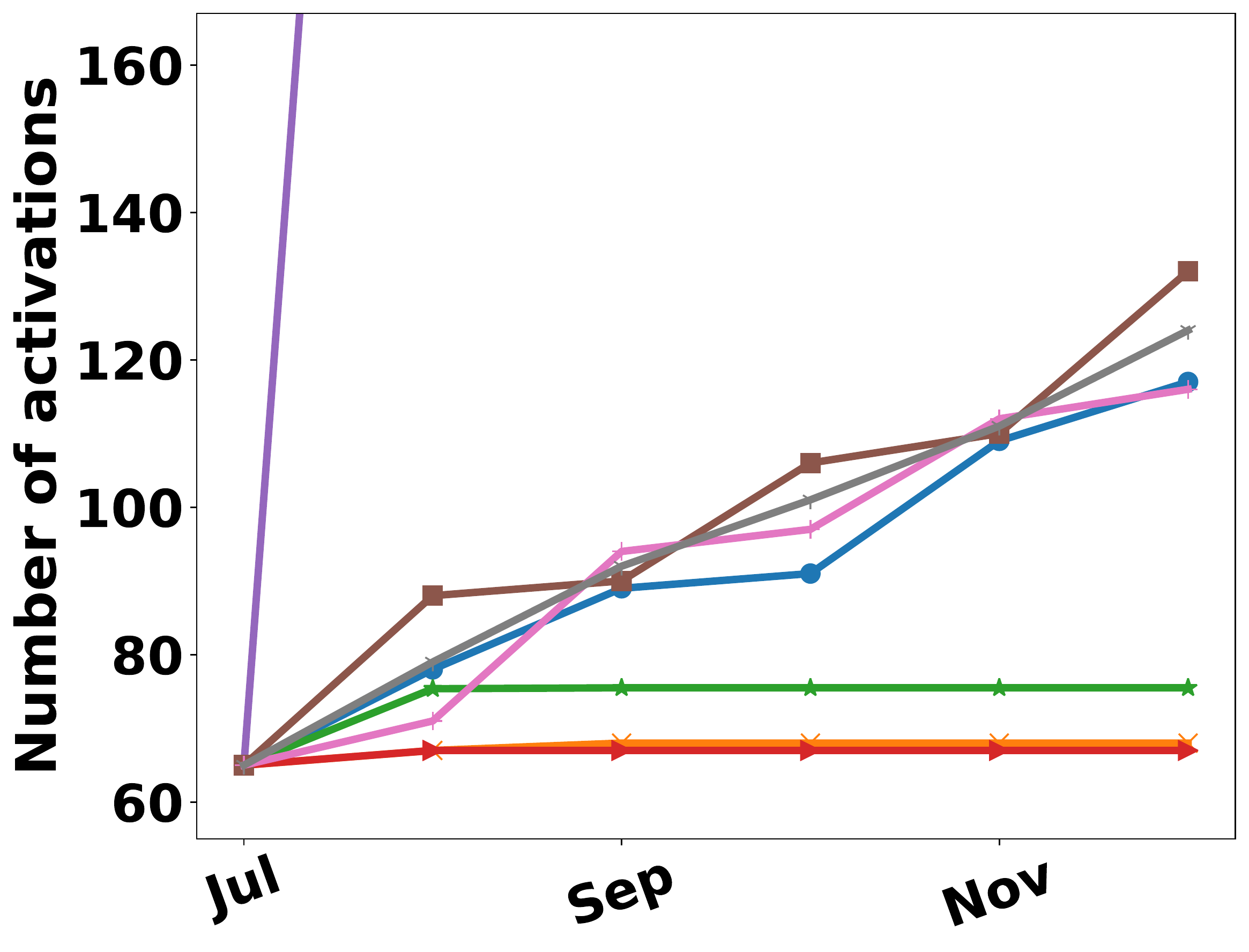}
        \caption{}
        \label{fig:hateful_compare_2016-6}
    \end{subfigure}
    
    \begin{subfigure}[t]{\subfigurefirst}
        \includegraphics[width=\textwidth]{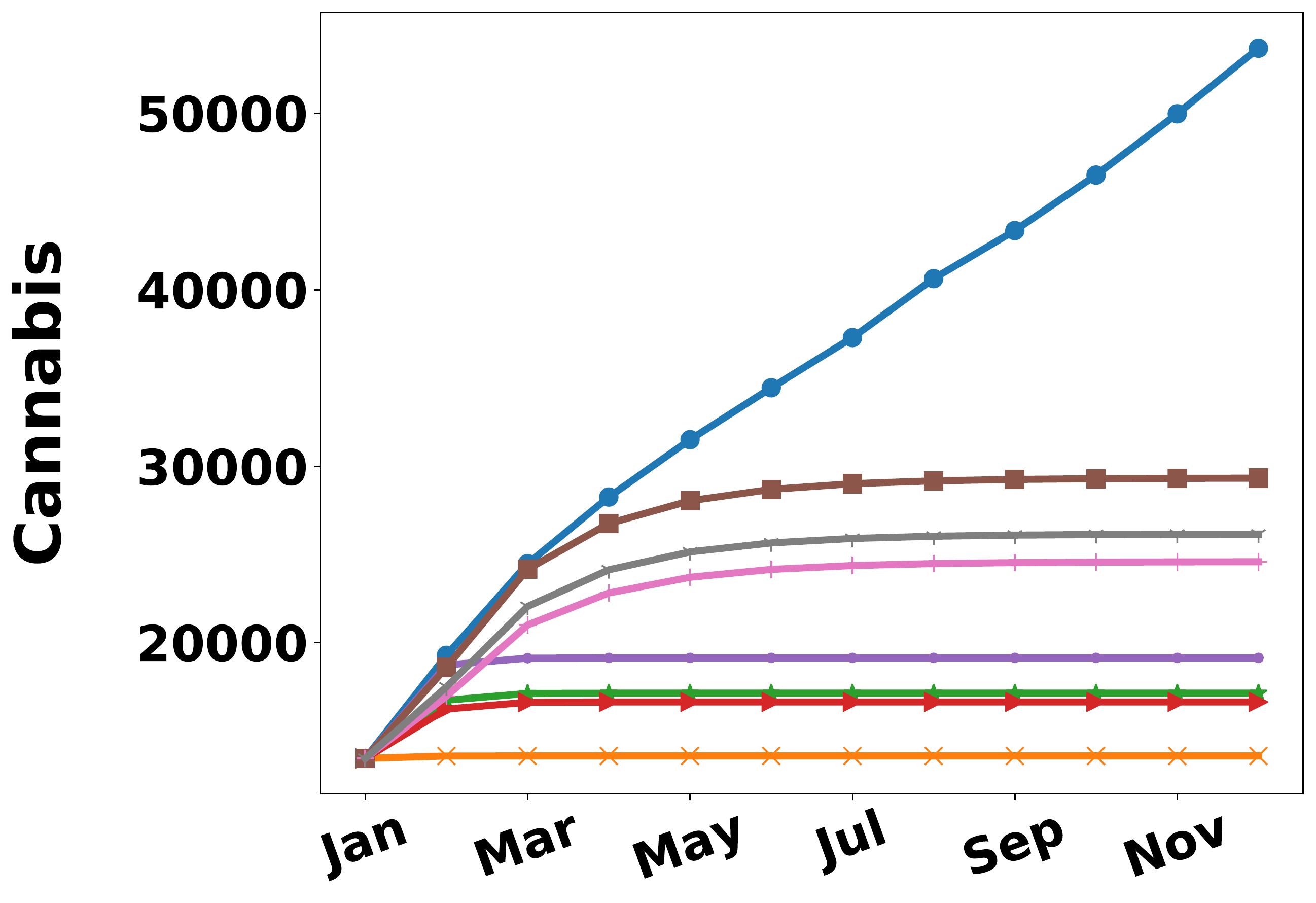}
        \caption{}
        \label{fig:cannabis_compare_0}
    \end{subfigure}
    \begin{subfigure}[t]{\subfiguresize}
        \includegraphics[width=\textwidth]{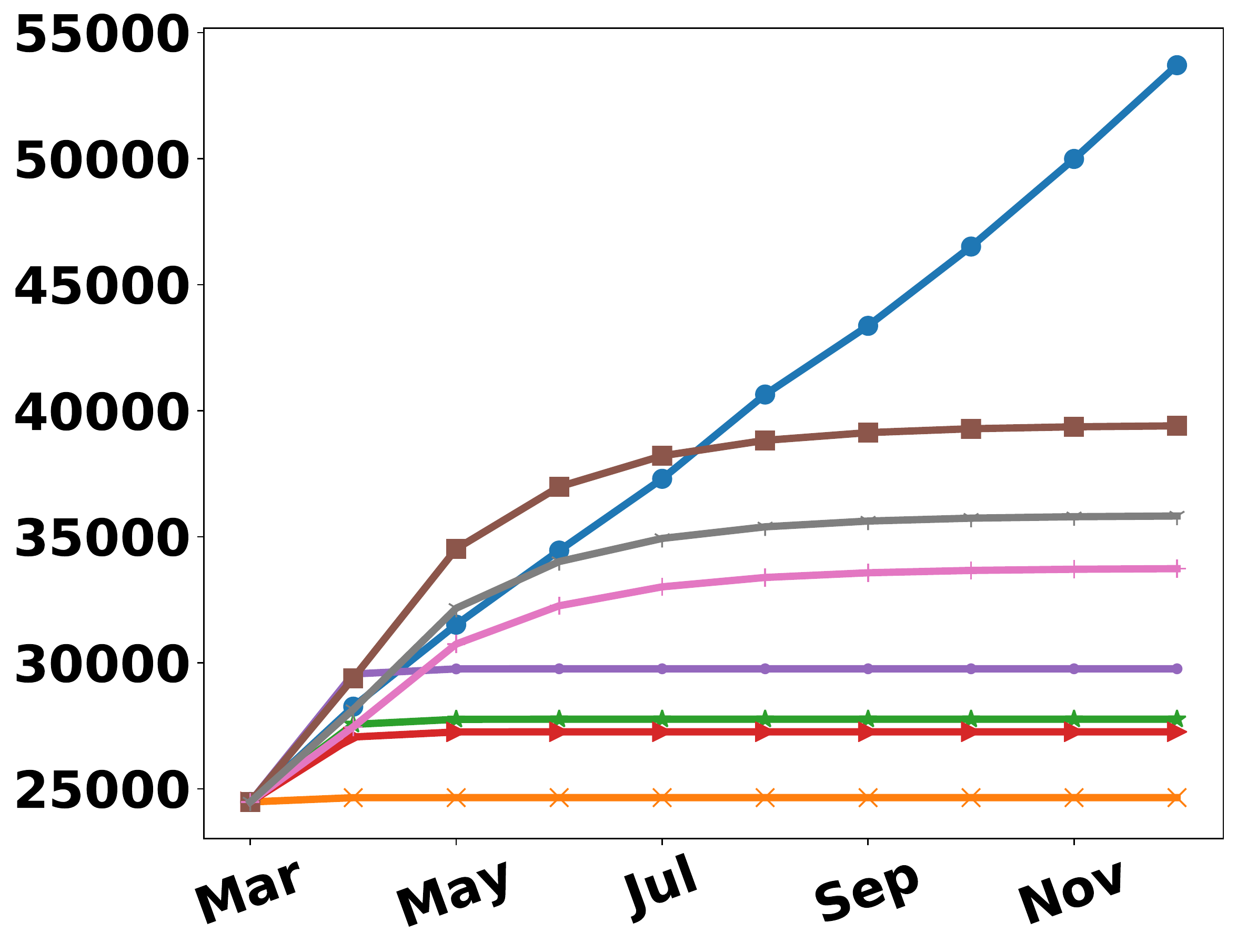}
        \caption{}
        \label{fig:cannabis_compare_2}
    \end{subfigure}
    \begin{subfigure}[t]{\subfiguresize}
        \includegraphics[width=\textwidth]{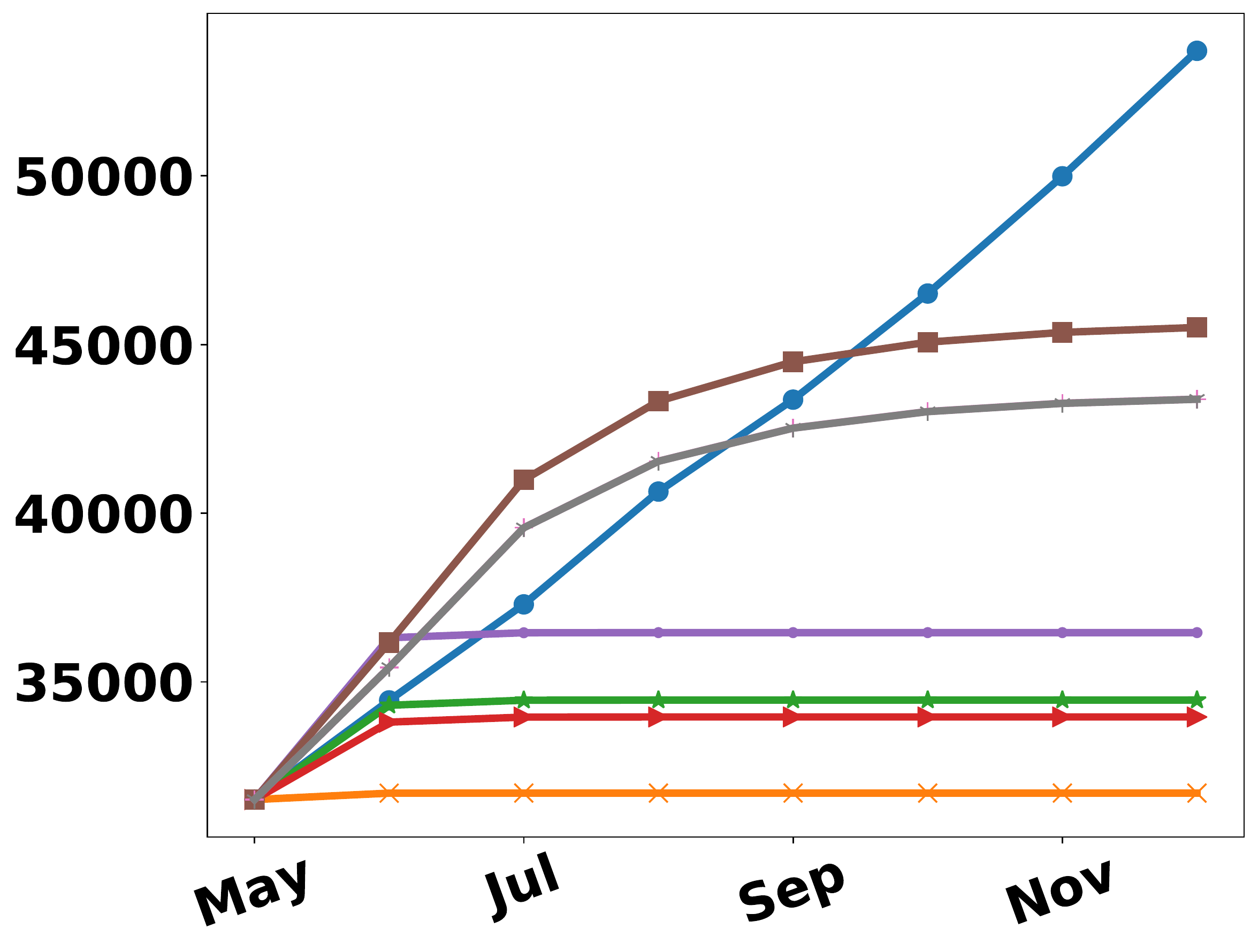}
        \caption{}
        \label{fig:cannabis_compare_4}
    \end{subfigure}
    \begin{subfigure}[t]{\subfiguresize}
        \includegraphics[width=\textwidth]{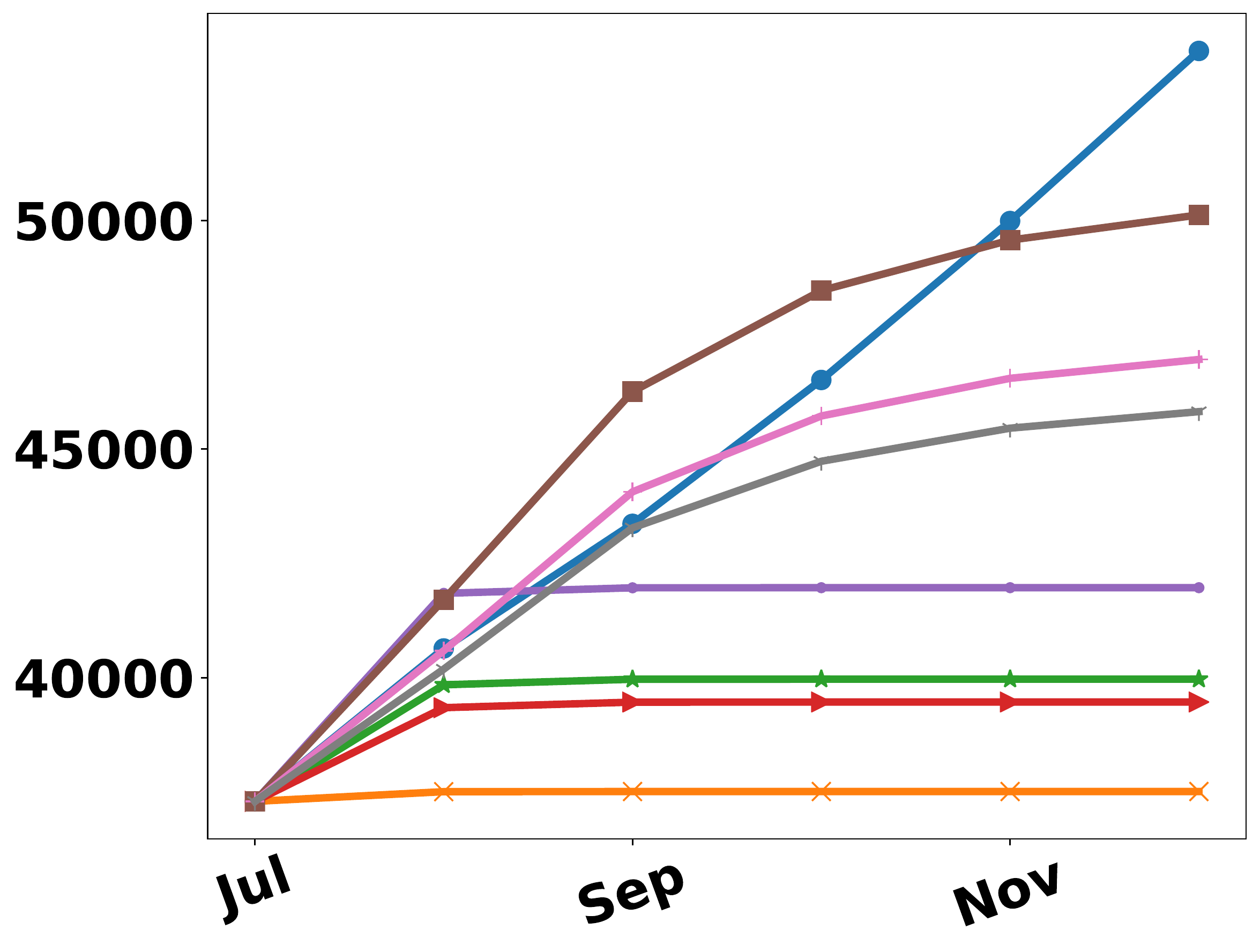}
        \caption{}
        \label{fig:cannabis_compare_6}
    \end{subfigure}
    
    \begin{subfigure}[t]{\subfigurefirst}
        \includegraphics[width=\textwidth]{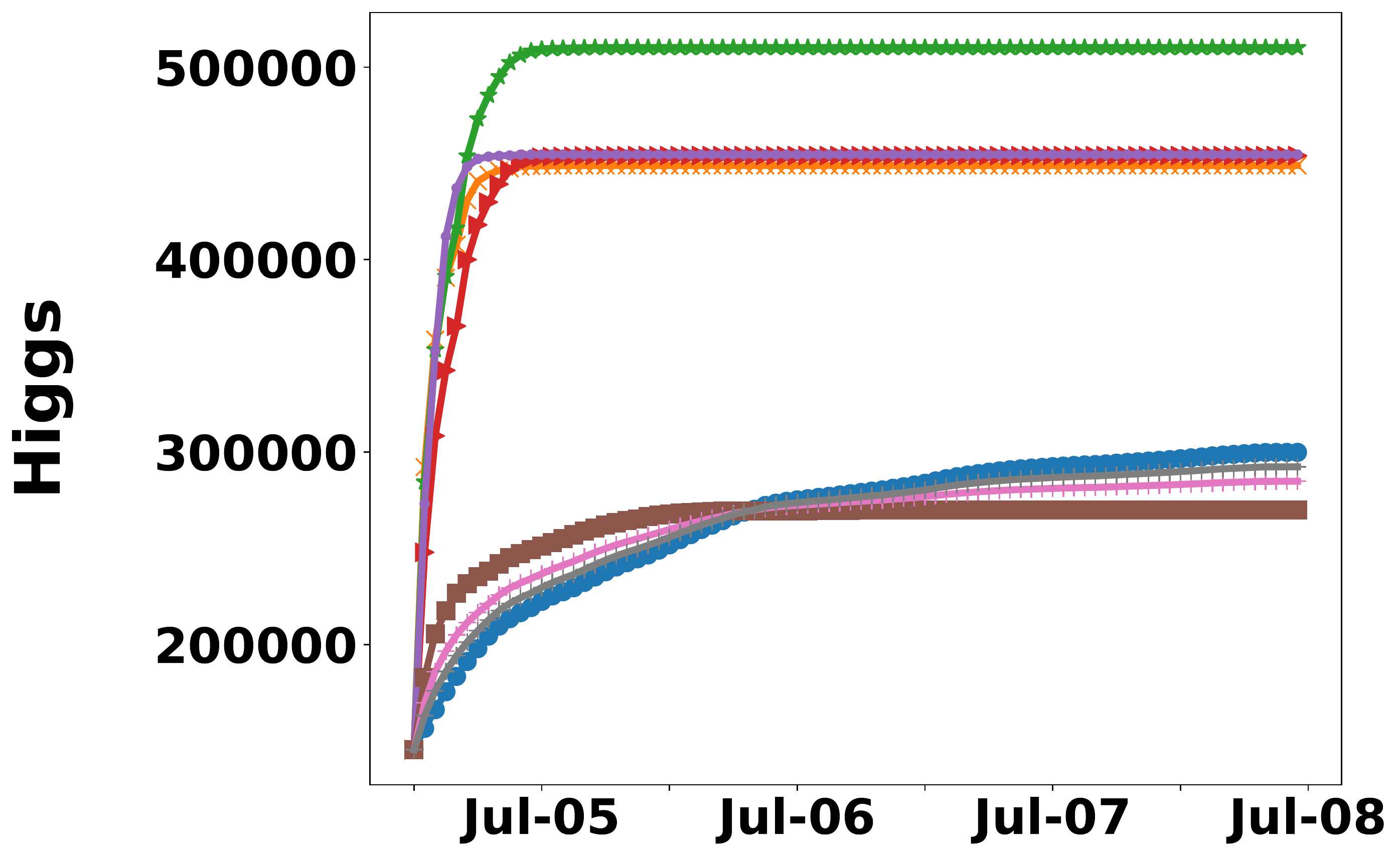}
        \caption{}
        \label{fig:higgs_compare_0}
    \end{subfigure}
    \begin{subfigure}[t]{\subfiguresize}
        \includegraphics[width=\textwidth]{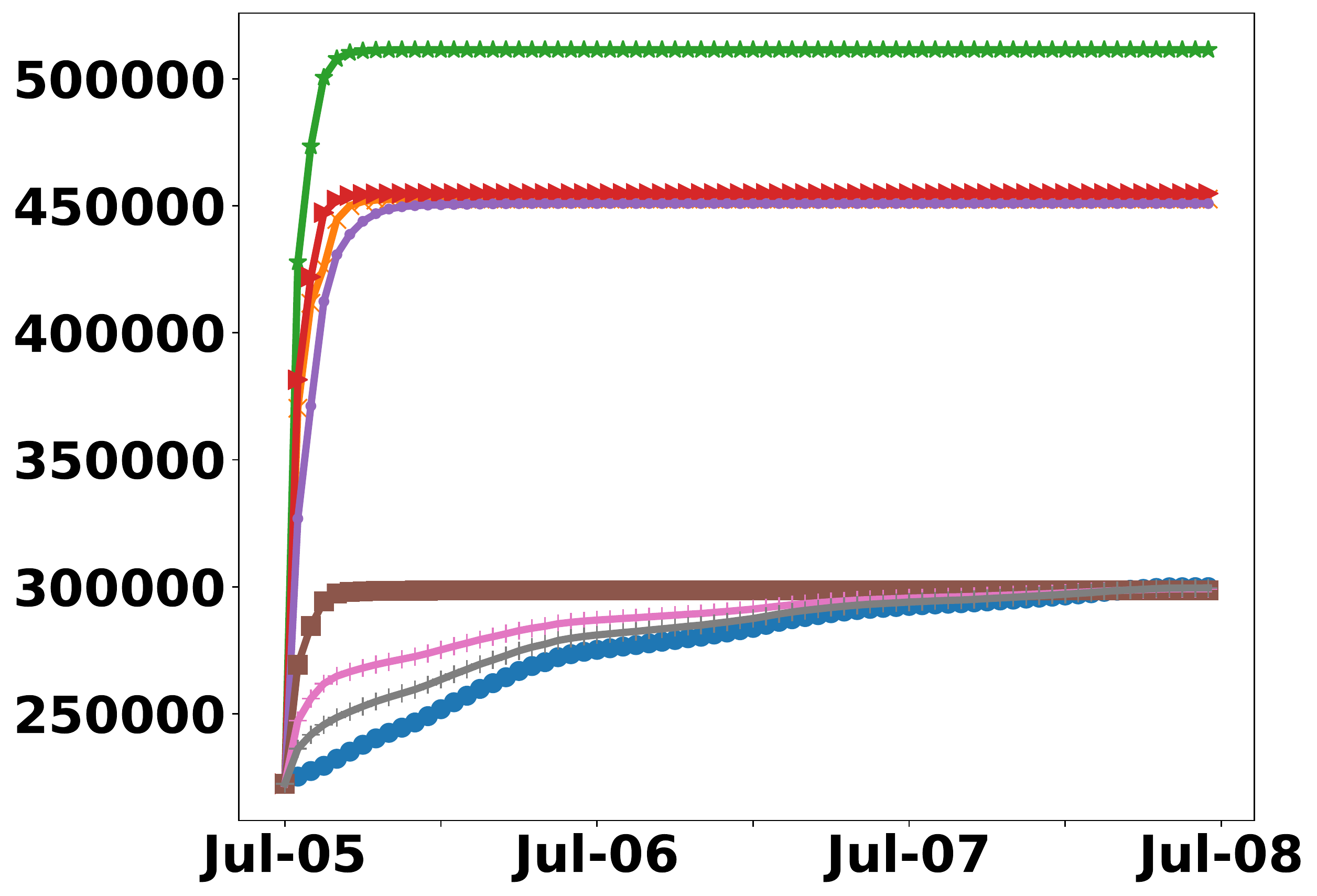}
        \caption{}
        \label{fig:higgs_compare_2}
    \end{subfigure}
    \begin{subfigure}[t]{\subfiguresize}
        \includegraphics[width=\textwidth]{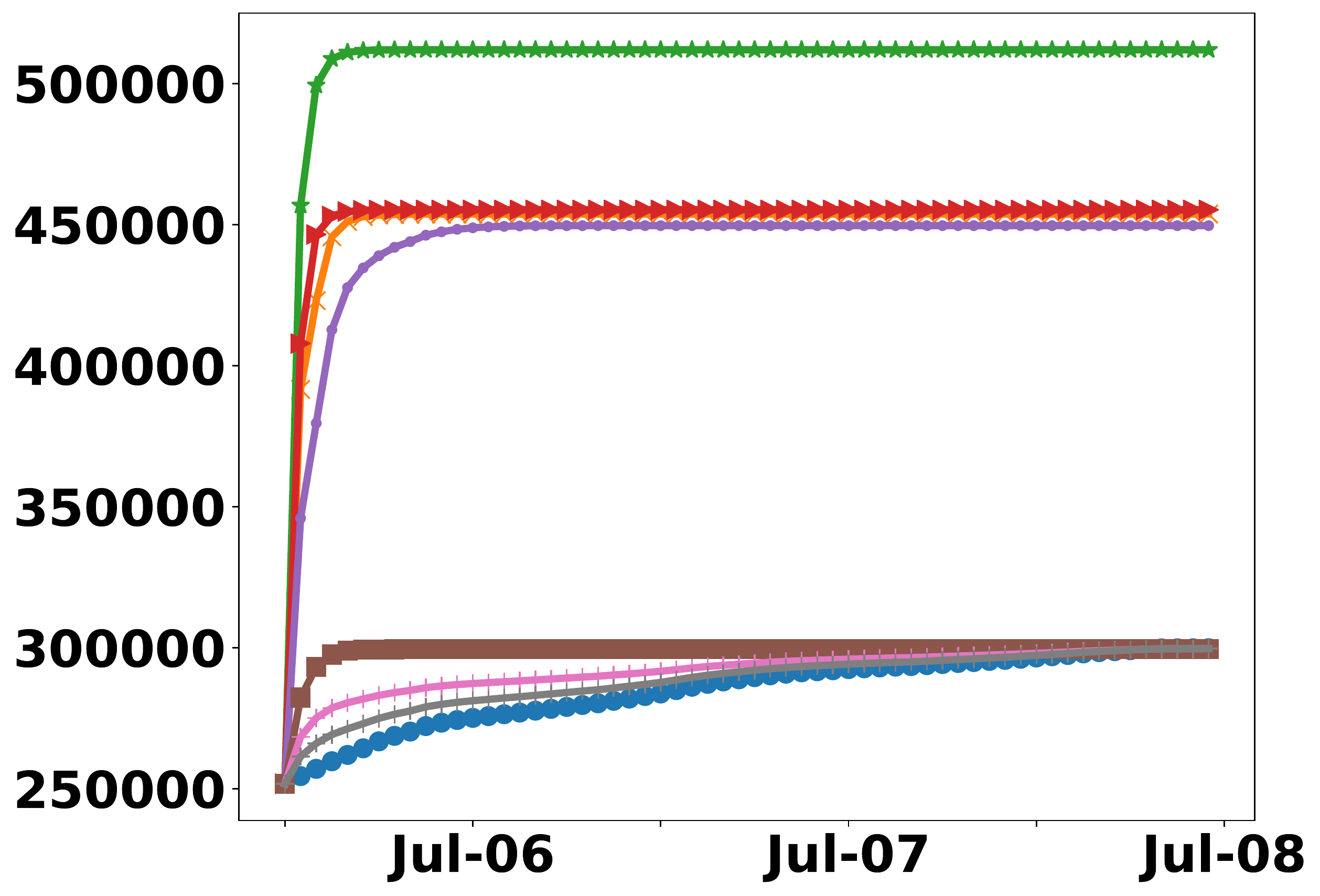}
        \caption{}
        \label{fig:higgs_compare_4}
    \end{subfigure}
    \begin{subfigure}[t]{\subfiguresize}
        \includegraphics[width=\textwidth]{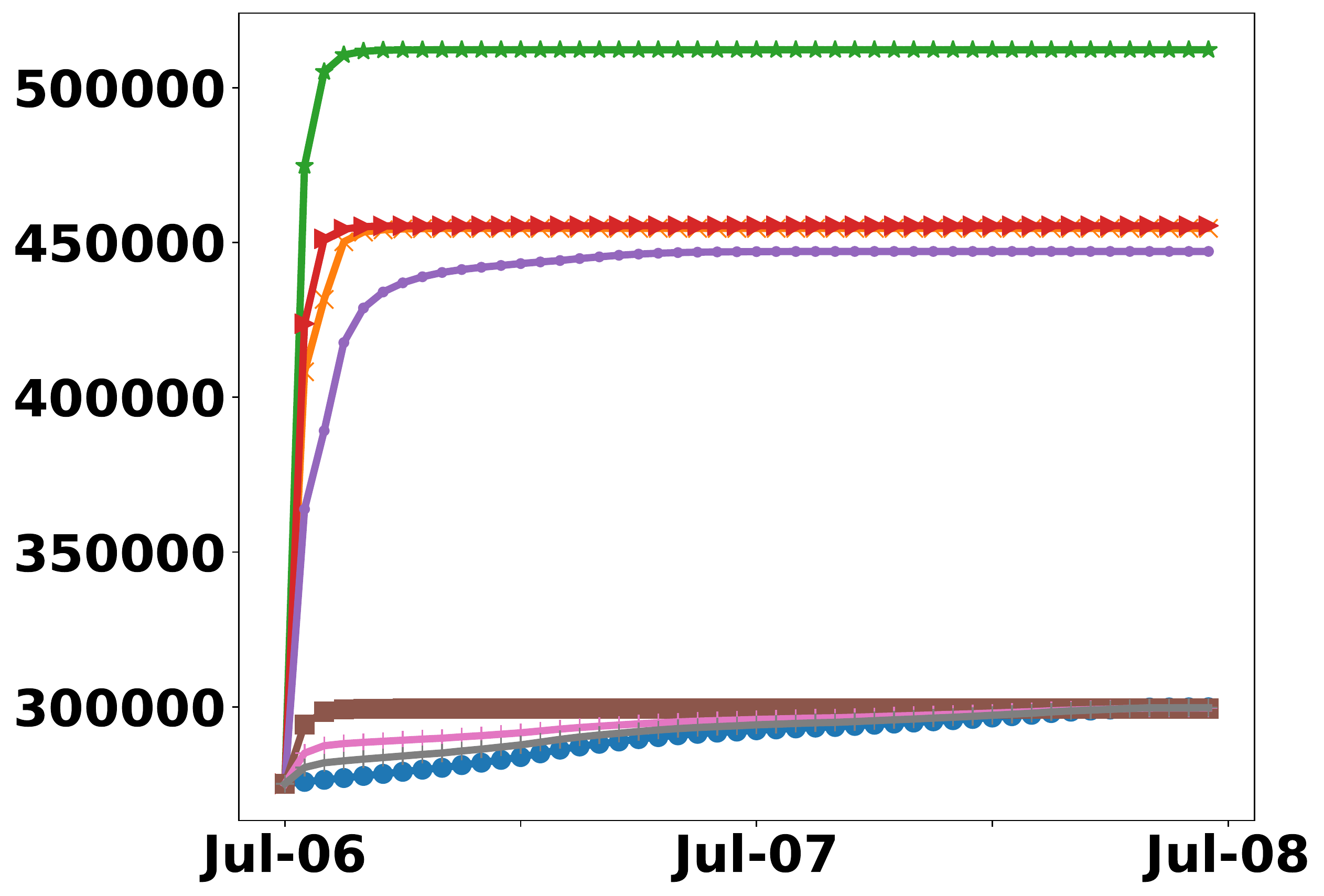}
        \caption{}
        \label{fig:higgs_compare_6}
    \end{subfigure}
    \caption{
    Comparison of diffusion size prediction on three real world datasets. 
    Our models have the closest estimation of reach over longer time periods whereas the baselines incorrectly predict diffusion saturation in the early stages. 
    }\label{fig:real_compare_appendix}
\end{figure*}
}

\subsection{Task 1 results: node threshold prediction}

Since real-world datasets do not have true node thresholds, we present results on the synthetic datasets for this task.
\noappendix{
Figure~\ref{fig:synthetic_mse_appendix} shows the MSE for each threshold estimation method varied across network generation models and for the \emph{Linear} and \emph{Quadrant} threshold generation methods.
}
\yesappendix{
Figure~\ref{fig:synthetic_mse} shows the MSE for each threshold estimation method varied across network generation models for the \emph{Linear} threshold generation method.
\emph{Quadrant} dataset results are available in the Appendix and show ST-DT performing the best overall.
}
We see that the node threshold estimators perform better across setups compared to baseline methods that do not consider node threshold prediction.
Our methods also perform noticeably better than the Linear Regression baseline in forest fire, and Watts-Strogatz.
All individualized threshold prediction models perform noticeably better than models that do not consider node threshold prediction.

From the results, ST-DT performs the best overall on the \emph{Linear} setup.
Just estimating thresholds using a regression method, like Linear Regression, does not always perform better than other baselines.
For example, Linear Regression has very high error in Watts-Strogatz networks for smaller number of neighbors.
ST-DT also performs the best overall on the nonlinear threshold in the \emph{Quadrant} setup, followed by Causal Tree.

\subsection{Task 2 results: activated node prediction}

For task 2, we show results for both synthetic and real-world datasets.
Table~\ref{tab:acc} shows the average Jaccard index for all datasets.
Overall, our models are able to achieve the best average Jaccard index across all datasets which is especially pronounced in the real-world datasets.
ST-DT achieves the highest average Jaccard index in all synthetic datasets.
In the real-world datasets, we see that our models are able to achieve significantly higher average Jaccard index compared to the baseline methods.
One potential reason for the discrepancy between the synthetic and real-world datasets is that synthetic datasets have generated features and a specific linear or quadrant threshold function.
On the other hand, our formulation and estimation using individual level features can better capture correct node thresholds for LTM for more accurate prediction in real-world scenarios.

\subsection{Task 3 results: diffusion size prediction}

\yesappendix{
We show results on diffusion size prediction on one network snapshot for real-world datasets in Figure~\ref{fig:real_compare}.
}
\noappendix{
We show results on diffusion size prediction on one network snapshot for real-world datasets in Figure~\ref{fig:real_compare_appendix}.
}
A good model should have a curve that is close to the true diffusion curve (in blue).
\yesappendix{
Other snapshots can be found in the Appendix and are consistent with the results here.
}

\textbf{Hateful Users dataset.}
\yesappendix{
Figure~\ref{fig:hateful_main} shows the diffusion size predictions for the Hateful Users dataset from Jan 2016 to Dec 2016.
}
\noappendix{
Figures~\ref{fig:hateful_compare_2016-0}-~\ref{fig:hateful_compare_2016-6} show the diffusion size predictions for the Hateful Users dataset from Jan 2016 to Dec 2016, with various snapshots (snapshots for Jan, Mar, May, Jul).
}
Our models initially overestimate reach, but are able to predict close to the final reach.
\yesappendix{
With more training data they estimate reach more accurately.
}
\noappendix{
With more data, our models estimate reach more accurately (e.g., Figure~\ref{fig:hateful_compare_2016-0} v.s.\ Figure~\ref{fig:hateful_compare_2016-6}).
}
In this dataset, Linear Regression significantly overestimates diffusion prediction in all snapshots.
\yesappendix{
In the 2017 dataset we notice the same trends as in the 2016 dataset, so we omit the plots.
}

\textbf{Cannabis dataset.}    
\yesappendix{
Figure~\ref{fig:cannabis_main} shows diffusion size predictions on the Cannabis dataset from Jan 2017 to Dec 2017.
}
\noappendix{
Figures~\ref{fig:cannabis_compare_0}-~\ref{fig:cannabis_compare_6} show diffusion size predictions on the Cannabis dataset from Jan 2017 to Dec 2017. 
}
In contrast to the Hateful Users dataset, all models predict reach that saturates after a number of time steps.
One reason for this could be the sparsity of edges.
Additionally, different parts of the network may not necessarily interact with each other, which results in disjoint subgraphs.

\textbf{Higgs dataset.}
\yesappendix{
Figure~\ref{fig:higgs_main} shows the diffusion size predictions for the Higgs dataset.
}
\noappendix{
Figures~\ref{fig:higgs_compare_0}-~\ref{fig:higgs_compare_6} show the diffusion size predictions for the Higgs dataset.
}
Specifically, we start at noon, July 4th and increase the starting snapshots by 12 hours each time.
We see that the ST-DT performs the best overall among our methods, and our methods perform significantly better than baselines.
Additionally, the baselines significantly overestimate the reach predictions.

\section{Conclusion}
\label{sec:conclusion}

In this work, we proposed a causal inference approach for estimating node thresholds in the Linear Threshold Model (LTM).
We defined a new concept of heterogeneous peer effect estimation, and developed a structural causal model for LTM to identify and estimate peer effects.
We developed a new meta-learner, the ST-Learner, and adapted trigger-based causal trees to solve the threshold estimation problem through heterogeneous peer effects.
Our results on synthetic and real-world datasets showed our models are able to estimate individualized thresholds from data better than baseline methods, and can produce more accurate sets of activated nodes and diffusion size predictions in the context of LTM, especially for real-world data.
A fruitful avenue of research would be to develop models that estimate edge influence weights, or to jointly learn edge influence weights and thresholds through a causal inference lens.

\section{Acknowledgments}
This material is based on research sponsored in part by the Defense Advanced Research Projects Agency (DAPRA) under contract numbers HR00111990114 and HR001121C0168 and the National Science Foundation under grant No. 2047899. The views and conclusions contained herein are those of the authors and should not be interpreted as necessarily representing the official policies, either expressed or implied, of DARPA or the U.S. Government. The U.S. Government is authorized to reproduce and distribute reprints for governmental purposes notwithstanding any copyright annotation therein.

\bibliography{main.bib}

\section{Appendix}

\section{Individual threshold estimation for LTM}

The implication of mapping the node threshold estimation to a trigger-based heterogeneous treatment effect estimation problem is that an accurate trigger correctly estimates the node threshold.

\begingroup
\def\thetheorem{\ref{thm:trigger_theorem}}
\begin{theorem}
    For any node $v \in \bm{V}$, let $\threshold_v$ be $v$'s true threshold. Then $\hat{\threshold}_v$ that maximizes the CAPE with a trigger in eq~\eqref{eq:trigger_cate}:
    \begin{equation}
        \argmax_{\hat{\theta}_v} E[\mathcal{Y}(\Influence^t_v \geq \hat{\threshold}_v) - \mathcal{Y}(\Influence^t_v < \hat{\threshold}_v) \mid \mathbf{X}_v, Z_v],
    \end{equation}
    provides the best estimate of the node threshold, $\threshold_v$.
\end{theorem}
\addtocounter{theorem}{-1}
\endgroup

\begin{proof}
    Suppose $|N(v)|=n$ and $w_{uv}$ be an arbitrary weight such that $\sum_{u} w_{uv} = 1$.
	Let the true threshold of node $v$ be $\theta_v$.
    Define $\mathbf{A_v} = \mathbf{a_i}$ to be an assignment of activations of neighbors of $v$ and $W_v(\alpha_i)$ to be the outcome for the activated set $\mathbf{a_i}$:
    \begin{align}
		W_v(\alpha) & = 
    	\begin{cases}
    		0 & \text{if $\Influence_v(\mathbf{a_i}) < \theta_v$}, \\
            1 & \text{if $\Influence_v(\mathbf{a_i}) \geq \theta_v$}.
		\end{cases}
    \end{align}
    Where $\Influence_v$ is the activation influence. 
    Let the set of potential outcomes, $\mathcal{Y}_v$, below and above the estimated trigger $\hat{\theta_v}$ be:
	\begin{align}
		\mathcal{Y}_v(\Influence_v < \hat{\theta}_v) & = \{ W_v(\mathbf{a_i}) \mid \Influence_v(\mathbf{a_i}) < \hat{\theta_v} \}
		\\
		\mathcal{Y}_v(\Influence_v \geq \hat{\theta}_v) & = \{ W_v(\mathbf{a_i}) \mid \Influence_v(\mathbf{a_i}) \geq \hat{\theta}_v \}.
    \end{align}
    Let $N_0$ and $N_1$ be the size of the sets $\mathcal{Y}_v(\Influence_v < \hat{\theta}_v)$ and $\mathcal{Y}_v(\Influence_v \geq \hat{\theta}_v)$, respectively.
    Define the set of outcomes when node $v$ is not activated and activated with the true threshold as $\mathcal{Z}_v(0)$ and $\mathcal{Z}_v(1)$:
    \begin{align}
        \mathcal{Z}_v(0) &= \{W_v(\mathbf{a_i}) \mid \Influence_v(\mathbf{a_i}) < \theta_v \},
        \\
        \mathcal{Z}_v(1) &= \{W_v(\mathbf{a_i}) \mid \Influence_v(\mathbf{a_i}) \geq \theta_v \}.
    \end{align}
    Let $N_{\mathcal{Z}_0}$ and $N_{\mathcal{Z}_1}$ be the cardinalities of set $\mathcal{Z}_v(0)$ and $\mathcal{Z}_v(1)$, respectively.
    The difference between $\mathcal{Y}_v$ and $\mathcal{Z}_v$ is that $\mathcal{Y}_v$ contains the true potential activations based on the \emph{estimated} threshold and $\mathcal{Z}_v$ contains the true potential activations based on the \emph{true} threshold.
    We define expectation over the sets $\mathcal{Y}_v$ and $\mathcal{Z}_v$ as the mean over the set.
    There are three cases for the estimated threshold, $\hat{\theta}_v$:
    \begin{case}[$\hat{\theta}_v < \theta_v$]
        We compute the expected outcomes below and above the estimated trigger $\hat{\theta}_v$:
        \begin{align}
            E[\mathcal{Y}_v(\Influence_v < \hat{\theta}_v)]  & = E[\{ W_v(\mathbf{a_i}) \mid \Influence_v(\mathbf{a_i}) < \hat{\theta}_v\}] \nonumber
            \\
            &= \frac{1}{N_0} \sum_{i} W_v(\mathbf{a_i}) = 0,
        \end{align}
        since there are no times when $W_v$ will be $1$ (activated) since all outcomes below the estimated threshold $\hat{\theta}_v$ are also below the true threshold $\theta_v$.
        The expected value above the estimated threshold is:
		\begin{gather}
            E[\mathcal{Y}_v(\Influence_v  \geq \hat{\theta}_v)]   = E[\{ W_v(\mathbf{a_i}) \mid \Influence_v(\mathbf{a_i}) \geq \hat{\theta}_v\}] \nonumber
            \\
            = \frac{1}{N_1} \sum_i W_v(\mathbf{a_i}) = \frac{1}{N_1} \sum_i W_v(\mathbf{a_i}) + \frac{1}{N_1} \sum_j W_v(\mathbf{a_j}) \nonumber
            \\
            = \frac{1}{N_1} \cdot N_{\mathcal{Z}_1} = \frac{N_{\mathcal{Z}_1}}{N_1}.
        \end{gather}
        The effect is the difference above and below the estimated trigger:
        \begin{align}
            E[\mathcal{Y}_v(\Influence_v \geq \hat{\theta_v}) - \mathcal{Y}_v(\Influence_v < \hat{\theta_v})] & = \frac{N_{\mathcal{Z}_1}}{N_1} - 0 = \frac{N_{\mathcal{Z}_1}}{N_1} < 1.
        \end{align}
        Note that $N_{\mathcal{Z}_1} < N_1$ since $N_1$ has cases above the estimated threshold, which is smaller than the true threshold.
    \end{case}
    \begin{case}[$\hat{\theta}_v > \theta_v$]
        \begin{align}
            E[\mathcal{Y}_v(\Influence_v < \hat{\theta}_v)]  & = E[\{ W_v(\mathbf{a_i}) \mid \Influence_v(\mathbf{a_i}) < \hat{\theta}_v\}]
            \nonumber \\
            & = \frac{1}{N_0} \sum_{i} W_v(\mathbf{a_i}) + \frac{1}{N_0} \sum_{j} W_v(\mathbf{a_i})
            \nonumber \\
            & = 0 + \frac{1}{N_0} \cdot (N_0 - N_{\mathcal{Z}_0}) \nonumber \\
            & = \frac{N_0 - N_{\mathcal{Z}_0}}{N_0}.
        \end{align}
        \begin{align}
            E[\mathcal{Y}_v(\Influence_v \geq \hat{\theta}_v)]  & = E[\{ W_v(\mathbf{a_i}) \mid \Influence_v(\mathbf{a_i}) \geq \hat{\theta}_v\}] \nonumber \\ 
            & = \frac{1}{N_1} \cdot N_1 = 1.
        \end{align}
        \begin{align}
            E[\mathcal{Y}_v(\Influence_v \geq \hat{\theta_v}) & - \mathcal{Y}_v(\Influence_v < \hat{\theta_v})] = 1 - \frac{N_0 - N_{\mathcal{Z}_0}}{N_0} < 1.
        \end{align}
    \end{case}
    \begin{case}[$\hat{\theta}_v = \theta_v$]
		\begin{align}
            E[\mathcal{Y}_v(\Influence_v < \hat{\theta}_v)]  & = E[\{ W_v(\mathbf{a_i}) \mid \Influence_v(\mathbf{a_i}) < \hat{\theta}_v\}]
            \nonumber\\
			& = \frac{1}{N_0} \sum_i W_v(\mathbf{a_i}) = 0.
			\\
			E[\mathcal{Y}_v(\Influence_v \geq \hat{\theta}_v)]  & = E[\{ W_v(\mathbf{a_i}) \mid \Influence_v(\mathbf{a_i}) \geq \hat{\theta}_v\}]
            \nonumber \\
            & = \frac{1}{N_1} \sum_i W_v(\mathbf{a_i}) = 1.
			\\
			E[\mathcal{Y}_v(\Influence_v \geq \hat{\theta_v}) & - \mathcal{Y}_v(\Influence_v < \hat{\theta_v})] = 1 - 0 = 1.
        \end{align}
        The maximum causal effect only occurs when $\hat{\theta}_v = \theta_v$.
		Therefore if we can estimate the trigger that maximizes CAPE we find the true node threshold.
		As a result, we can use any trigger-based HTE estimation method to estimate node thresholds.
		\qedhere
    \end{case}
\end{proof}

\subsection{Threshold estimation algorithms}

\subsubsection{Trigger-based causal trees}

Here we describe the Causal Tree (CT-HV) algorithm proposed in~\cite{tran-aaai2019}. Let $\mathbf{X}^{\ell}$ be the features in partition $\ell$, which represents a node in the tree which contains $N_\ell$ samples, and let $\hat{\mu}_1(\ell)$ and $\hat{\mu}_0(\ell)$ be the mean of outcomes when treated and non-treated.
The estimate of CATE of any partition is $\hat{\tau_c}(\mathbf{X}^\ell) = \hat{\mu}_1(\ell) - \hat{\mu}_0(\ell)$.
Given a partition $\ell$ that needs to be partitioned further into two children $\ell_1, \ell_2$, the causal tree algorithm finds the split on features that maximizes the weighted CATE in each child:
\begin{equation}
    \maxl_{\ell_1, \ell_2} N_{\ell_1} \cdot \hat{\tau_c}(\mathbf{X}^{\ell_1}) +  N_{\ell_2} \cdot \hat{\tau_c}(\mathbf{X}^{\ell_2}).
\end{equation}
We refer to the two quantities $N_{\ell_1} \cdot \hat{\tau_c}(\mathbf{X}^{\ell_1})$ and $N_{\ell_2} \cdot \hat{\tau_c}(\mathbf{X}^{\ell_2})$ as partition measures for $\ell_1$ and $\ell_2$.
In our work, we adapt the CT-HV algorithm~\cite{tran-aaai2019} for the problem of node threshold estimation.

In order to learn triggers, an additional search is done at each split to find the trigger that maximizes the effect estimation in each split.
Let $F(\ell)$ represent the partition measure for CT-HV and let $\hat{\theta_v}$ be some estimated trigger for partition $\ell$.
Define $M_1(\ell, \hat{\theta_v})$ and $M_0(\ell, \hat{\theta_v})$ to be the expected mean outcomes when above and below the trigger $\hat{\theta_v}$, and $F(\ell, \hat{\theta_v})$ be the partition measure with trigger $\hat{\theta_v}$: $F(\ell, \hat{\theta_v}) = M_1(\ell, \hat{\theta_v}) - M_0(\ell, \hat{\theta_v}) = E[\mathcal{Y}_\ell(\Influence \geq \hat{\theta_v}) - \mathcal{Y}_\ell(\Influence < \hat{\theta_v})].$
Then to find a trigger in partition $\ell$, we find the trigger that maximizes the partition measure: \( F(\ell, \hat{\theta_v}) \).
This results in a unique trigger in each partition.

\subsection{ST-Learner}

Pseudocode is provided for ST-Learner in Algorithm~\ref{alg:stlearner}.
Both the training and prediction subroutine are included.
In the train subroutine, we train a base learner, \( f \) to predict the outcome \( Y \) using node features \( \mathbf{X} \) and the activation influence \( I \).
To predict a trigger for a new test example, we compute and score the predicted outcome for all treatment values found in the dataset.
Then, for each potential trigger to consider, we compute an average above and below that trigger, and find the trigger that maximizes the difference in outcomes.
The additional complexity added by the search for triggers is on the order of potential treatment values $|B|$.
For example, suppose a base learner of Linear Regression requires $O(d)$ time for prediction, then ST-Learner with Linear Regression will take $O(d \cdot |B|)$ time.

\begin{algorithm}[t]
\small
\caption{ST-Learner}\label{alg:stlearner}
  \begin{algorithmic}[1]
    \Require Training set $(\mathbf{X}, Y, \Influence)$, test example $\mathbf{x}^{te}$, base learner $f$, possible treatments $\Alltreatments$, possible triggers $\Theta$
    \Ensure The causal effect estimate for $\mathbf{x}^{te}$
    
    \Function{Train}{$\mathbf{X}, Y, \Influence$}
        \State Fit $f$ to predict $Y$, $f$: $(\mathbf{X}, I) \rightarrow Y$
        \State $\implies f(\mathbf{X}, \Influence) = E[Y \mid \mathbf{X}, \Influence]$
    \EndFunction
    
    \Function{Predict}{$\mathbf{x}^{te}$}
        \For{each $\alltreatments_i$ in $\Alltreatments$}
            \State Compute and store $f(\mathbf{x}^{te}, \alltreatments_i)$
        \EndFor
        \State max\_trigger $= 0$, max\_effect $= 0$
        \For{each $r_i$ in $\Theta$}
            \State Compute $t_1 = \frac{1}{|\Theta_i^1|} \sum_{\alltreatments_k \geq r_i} f(\mathbf{x}^{te}, \alltreatments_k)$
            \State Compute $t_0 = \frac{1}{|\Theta_i^0|} \sum_{\alltreatments_j < r_i} f(\mathbf{x}^{te}, \alltreatments_j)$
            \State $e = t_1 - t_0$
            \If{$e >$ max\_effect}
                \State max\_trigger $= r_i$, max\_effect $= e$
            \EndIf
        \EndFor
        \State \textbf{return} max\_trigger, max\_effect
    \EndFunction
  \end{algorithmic}
\end{algorithm}

\subsection{Additional results}

\subsubsection{Task 1: node threshold prediction}

\begin{figure*}[!ht]
	\centering
    \begin{subfigure}{\legendsize}
        \includegraphics[width=\textwidth]{figures/legend_cut.pdf}    
    \end{subfigure}

	\begin{subfigure}[b]{\subfiguresize}
		\includegraphics[width=\textwidth]{figures/synthetic/linear/erdos-renyi.pdf}
		\caption{Erdos-Renyi}\label{fig:linear-erdos}
	\end{subfigure}
	\begin{subfigure}[b]{\subfiguresize}
		\includegraphics[width=\textwidth]{figures/synthetic/linear/pref-attach.pdf}
		\caption{Pref. attachment}\label{fig:linear-pref}
	\end{subfigure}
	\begin{subfigure}[b]{\subfiguresize}
		\includegraphics[width=\textwidth]{figures/synthetic/linear/forest-fire.pdf}
		\caption{Forest fire}\label{fig:linear-forest}
	\end{subfigure}
	\begin{subfigure}[b]{\subfiguresize}
		\includegraphics[width=\textwidth]{figures/synthetic/linear/small-world.pdf}
		\caption{Watts-Strogatz}\label{fig:linear-watts}
	\end{subfigure}

	\begin{subfigure}[b]{\subfiguresize}
		\includegraphics[width=\textwidth]{figures/synthetic/quad2/erdos-renyi.pdf}
		\caption{Erdos-Renyi}\label{fig:quad-erdos}
	\end{subfigure}
	\begin{subfigure}[b]{\subfiguresize}
		\includegraphics[width=\textwidth]{figures/synthetic/quad2/pref-attach.pdf}
		\caption{Pref. attachment}\label{fig:quad-pref}
	\end{subfigure}
	\begin{subfigure}[b]{\subfiguresize}
		\includegraphics[width=\textwidth]{figures/synthetic/quad2/forest-fire.pdf}
		\caption{Forest fire}\label{fig:quad-forest}
	\end{subfigure}
	\begin{subfigure}[b]{\subfiguresize}
		\includegraphics[width=\textwidth]{figures/synthetic/quad2/small-world.pdf}
		\caption{Watts-Strogatz}\label{fig:quad-watts}
	\end{subfigure}
	\caption{\kdd{MSE of threshold prediction for the \emph{Linear} (top row) and \emph{Quadrant} (bottom row) setups over different parameters.}}\label{fig:synthetic_mse_appendix}
\end{figure*}

Figure~\ref{fig:synthetic_mse_appendix} shows the change in MSE in the \emph{Linear} and \emph{Quadrant} setup.
From the results, we see that Causal Tree in general performs better on the nonlinear threshold in the \emph{Quadrant} dataset, compared to the linearly generated threshold.

\subsection{Task 3: diffusion size prediction}

Figure~\ref{fig:real_compare_appendix} shows additional figures for the real-world datasets. Each figure starts at a different network snapshot, where a snapshot is the current structure and activations at time \( t \).
For example, Figure~\ref{fig:hateful_compare_2016-2} learns on data starting at the snapshot at the end of March 2016.
A later snapshot means there are more activations in the dataset, and thus more training data.
As the training data size increases, our methods can achieve better diffusion size prediction.

\textbf{Hateful Users dataset.}
Figures~\ref{fig:hateful_compare_2016-0}-\ref{fig:hateful_compare_2016-6} show the diffusion predictions for the Hateful Users dataset from Jan 2016 to Dec 2016, with varying starting points (Jan, Mar, May, Jul).
We show diffusion simulations at four snapshots before Dec 2016.
Our models initially overestimate the diffusion size prediction
The first thing we notice is that our models initially overestimates the reach the diffusion prediction, but is able to predict close to the final diffusion amount.
Additionally, we see that as we get more time steps, our models obtains more accurate reach estimates.
For example, in Figures~\ref{fig:hateful_compare_2016-0} and~\ref{fig:hateful_compare_2016-2}, where we start with information in January and March, our models overestimates the prediction in the beginning. 
With more information, the reach estimates are better, as in Figures~\ref{fig:hateful_compare_2016-4} (starting in May), and~\ref{fig:hateful_compare_2016-6} (starting in July).
In the 2017 dataset, we notice the same trends as in the 2016 dataset: our models overestimates the reach prediction in the beginning, but predicts close to the final true amount, so we omit the plots.

\textbf{Cannabis dataset.}
Figures~\ref{fig:cannabis_compare_0}-\ref{fig:cannabis_compare_6} show results for threshold estimation on the Cannabis dataset from Jan 2017 to Dec 2017.
Contrast to the Hateful Users dataset, all models predict diffusion that saturate after a number of time steps.
Our models predict saturation closer to the final diffusion amount with more training data.

\textbf{Higgs dataset.}
Figure~\ref{fig:higgs_compare_0} shows the diffusion size predictions for the Higgs dataset.
Specifically, we start at noon, July 4th and increase the starting snapshots by 12 hours each time.
We see that the ST-Learner performs slightly better than Causal Tree and both outperform the baselines by a significant margin.
Additionally, the baselines significantly overestimate the reach predictions.

\end{document}